\documentclass[conference, 10pt]{IEEEtran}

\usepackage{local-macros}

\title{Algebraic Foundations of Proof Refinement}

\date{\today}
\author{
  \IEEEauthorblockN{Jonathan Sterling}
  \IEEEauthorblockA{Carnegie Mellon University}
  \and
  \IEEEauthorblockN{Robert Harper}
  \IEEEauthorblockA{Carnegie Mellon University}
}

\begin{document}

\maketitle

\begin{abstract}

  We contribute a general apparatus for \emph{dependent} tactic-based proof
  refinement in the \LCF{} tradition, in which the statements of subgoals may
  express a dependency on the proofs of other subgoals; this form of dependency
  is extremely useful and can serve as an \emph{algorithmic} alternative to
  extensions of \LCF{} based on non-local instantiation of schematic variables.
  Additionally, we introduce a novel behavioral distinction between
  refinement rules and tactics based on naturality.  Our
  framework, called \DependentLCF{}, is already deployed in the nascent
  \RedPRL{} proof assistant for computational cubical type theory.

\end{abstract}

\section{Introduction}

Interactive proof assistants are at their most basic level organized
around some form of proof refinement apparatus, which defines
the valid transitions between partial constructions, and induces a
notion of proof tactic. The proof refinement tradition begins
with Milner~et~al.'s \emph{Logic for Computable
  Functions}~\cite{milner:1972,gordon-milner-wadsworth:1979} and was
further developed in Cambridge LCF, HOL and
Isabelle~\cite{gordon:2000}, as well as the Nuprl
family~\cite{constable:1986,hickey:2001,redprl:2016};
tactic-based proof refinement is also used in the highly successful
Coq proof assistant~\cite{coq:reference-manual} as well as the new
Lean theorem prover~\cite{de-moura-kong-avigad-van-doorn-raumer:2015}.

\begin{notation}
  Throughout this paper, we employ a notational scheme where the
  active elements of mathematical statements or judgments are colored
  according to their \emph{mode}, i.e.\ whether they represent inputs
  to a mathematical construction or outputs. Terms in input-mode are
  colored with $\IMode{\textbf{blue}}$, whereas terms in output-mode are
  colored with $\OMode{\textbf{maroon}}$.
\end{notation}

\subsection{Proof refinement and evidence semantics}

At the heart of \LCF-style proof refinement is the coding of the
inference rules of a formal logic into partial functions which take a
conclusion and return a collection of premises (subgoals); this is
called backward inference. Often, this collection of subgoals
is equipped with a validation, a function that converts
evidence of the premises into evidence of the conclusion (which is
called forward inference).

An elementary example of the \LCF{} coding of inference rules can be
found in the right rule for conjunction in an intuitionistic sequent
calculus:
\begin{gather*}
  \IBox{
    \begin{array}{l}
      \infer[\mathfrak{\land_R}]{
      \IMode{\Gamma}\vdash\IMode{P\land Q}
      }{
        \deduce[\OMode{\mathcal{D}}]{\IMode{\Gamma}\vdash\IMode{P}}{}
        &
        \deduce[\OMode{\mathcal{E}}]{\IMode{\Gamma}\vdash\IMode{Q}}{}
      }
    \end{array}
  }
  \\
  \triangledown
  \\
  \OBox{
    \begin{RuleDefn}
      \RuleCase{\Gamma\vdash A \land B}{
        \left\langle
          \begin{array}{l}
            [\Gamma\vdash A, \Gamma\vdash B],\\
            \lambda[\mathcal{D},\mathcal{E}].\ \mathfrak{\land_R}(\mathcal{D}, \mathcal{E})
          \end{array}
        \right\rangle%
      }
      \\
      \RuleCase{\_}{\mathbf{raise}\ \mathsf{Fail}}
    \end{RuleDefn}
  }
\end{gather*}

In the example above, the validation produces exactly a proof of
$\IMode{\Gamma}\vdash\IMode{A\land B}$ in the intuitionistic sequent
calculus, but in general it is not necessary for the meaning of the
validations to correspond exactly with formal proofs in the logic
being modeled.

In both Edinburgh~LCF and Cambridge~LCF, the implementation of a proof
refinement logic was split between a trusted kernel which implemented
forward inference, and a separate module which implemented
backward inference (refinement) using the forward inference rules as
validations~\cite{paulson:1987}. Under this paradigm, every inference
rule must be implemented twice, once in each direction.

This approach of duplicating formal rules in forward and backward
inference is a \emph{particular} characteristic of the \LCF{}
revolution as it actually occurred, rather than a \emph{universal}
principle which is to be applied dogmatically in every application of
the \LCF{} methodology. On the contrary, it is possible to view a
collection of backward inference (refinement) rules as definitive in
its own right, independent of any forward inference characterization
of a logic; this insight, first achieved in the Nuprl proof
assistant~\cite{constable:1986}, enables a sharp distinction between
formal proof and evidence (the former corresponding to derivations and
the latter corresponding to realizers).

For instance, in the Nuprl family of proof assistants, the refinement
logic is a formal sequent calculus for a variant of Martin-L\"of's
type theory (see~\cite{martin-lof:1979}), but the validations produce
programs in an untyped language of realizers with computational
effects~\cite{rahli-bickford:2016}. Validations, rather than
duplicating the existing refinement rules in forward inference, form
the \enquote{program extraction} mechanism in Nuprl; this deep
integration of extraction into the refinement logic lies in stark
contrast with the purely extra-theoretic treatments of program
extraction in other proof assistants, such as Coq and Agda.

The ability for the notion of evidence to vary independently
from the notion of formal proof is a major strength of the
\LCF{} design: the latter is defined by the assignment of premises to
conclusions in trusted refinement rules, whereas the former is defined
in the validations of these rules.

Constable first made precise this idea under the name of
\emph{computational evidence
  semantics}~\cite{constable:1985,constable:2012}. In practice, this
technique of separating the proof theory from the evidence semantics
can be used to induce (sub)structural and behavioral invariants in
programs from arbitrary languages as they already exist,
regardless of whether these languages possess a sufficiently
proof-theoretic character (embodied, for instance, in the
much acclaimed decidable typing property).

This basic asymmetry between proof and computational evidence in
refinement logics corresponds closely with the origin of the sequent
calculus, which was designed as a theory of derivability for natural
deduction. More generally, it reflects the distinction between the
demonstration of a judgment and the construction it
effects~\cite{martin-lof:1987,prawitz:2012}.

\subsection{Attacking the dependency problem}\label{sec:barbarism}

An apparently essential part of the \LCF{} apparatus is that each
subgoal may be stated independently of evidence for the others; this
characteristic, originally identified as part of the
\emph{constructible subgoals} property by Harper in his
dissertation~\cite{harper:1985}, allows proof refinement to proceed in
parallel and in any order.

This restriction, however, raises difficulties in the definition of a
refinement logic for dependent type theory, where the statement of one
subgoal may depend on the evidence for another subgoal (a state of
affairs induced by families of propositions which are indexed in
proofs other another proposition). The most salient example of this
problem is given by the introduction rule for the dependent sum of a
family of types~\cite[p.~35]{harper:1985}. First, consider the
standard type membership rules, which pose no problem:
\[
  \infer[\mathfrak{\times^=_I}]{
    \IMode{\langle M, N\rangle} \in \IMode{(x : A)\times B[x]}
  }{
    \IMode{M}\in\IMode{A}
    &
    \IMode{N}\in\IMode{B[M]}
  }
\]

The premises can be stated purely on the basis of the conclusion,
because $M$ appears in the statement of the conclusion. However, if we
try to convert this to a refinement rule, in which we do not
have to specify the exact member $\langle M,N\rangle$ in advance, we
immediately run into trouble:
\[
  \infer[\mathfrak{\times_R}]{
    \deduce[\OMode{\langle M,N\rangle}]{\IMode{(x:A)\times B[x]}\ \mathit{true}}{}
  }{
    \deduce[\OMode{M}]{\IMode{A}\ \mathit{true}}{}
    &
    \deduce[\OMode{N}]{\IMode{B[M]}\ \mathit{true}}{}
  }
\]

Suspending for the moment one's suspicions about the above notation,
the fundamental problem is that this inference rule cannot be
translated into an \LCF{} rule, because the subgoal
$\IMode{B[M]}\ \mathit{true}$ cannot even be stated
until it is known what $M$ is, i.e.\ until we have somehow run the
validation for the completed proof of the first subgoal.

\subsubsection{Ignore the problem!}

The resolution adopted by the first several members of the Nuprl
family was to defer the solution of this problem to a later year,
thereby requiring that the user solve the first subgoal in advance of
applying the introduction rule. This amounted to defining a countable
family of rules $\mathfrak{\times_R}\{M\}$, with the term $M$ a
parameter:
\[
  \infer[\mathfrak{\times_R}\{M\}]{
    \deduce[\OMode{\langle M,N\rangle}]{\IMode{(x:A)\times B[x]}\ \mathit{true}}{}
  }{
    \deduce[\OMode{*}]{\IMode{M}\in\IMode{A}}{}
    &
    \deduce[\OMode{N}]{\IMode{B[M]}\ \mathit{true}}{}
  }
\]

The subgoals of the above rule are independent of each other, and it
can therefore easily be coded as an \LCF{} rule. However, this has
come at great cost, because the user of the proof assistant is no
longer able to produce the proof of $A$ using refinement, and must
choose some witness before proceeding. This is disruptive to the
interactive and incremental character of refinement proof, and we
believe that the time is ripe to revisit this question.

\subsubsection{Non-local unification}
The most commonly adopted solution is to introduce a notion of
existential variable and solve these non-locally in a proof via
unification or spontaneous resolution. The basic idea is that you proceed with the proof of the
second premise without yet knowing what $M$ is, and then at a certain
point, it will hopefully become obvious what it must be, on the
basis of some goal in the subtree of this premise:
\[
  \infer{
    \deduce[\OMode{\langle ?\mathfrak{m},N\rangle}]{\IMode{(x:A)\times B[x]}\ \mathit{true}}{}
  }{
    \deduce[\OMode{*}]{\IMode{?\mathfrak{m}}\in\IMode{A}}{}
    &
    \deduce[\OMode{N}]{\IMode{B[?\mathfrak{m}]}\ \mathit{true}}{}
  }
\]

An essential characteristic of this approach is that when it is known
what $?\mathfrak{m}$ must be, this knowledge propagates immediately to
all nodes in the proof that mention it. This is a fundamentally
imperative operation, and makes it very difficult to reason
about proof scripts; moreover, it complicates the development of a
clean and elegant semantics (denotational or otherwise) for proof
refinement. At the very least, a fully formal presentation of the
transition rules for such a system will be very difficult, both to
state and to understand.

Another potential problem with using unification to solve goals is
that one must be very cautious about how it is applied: using
unification uncritically in the refinement apparatus can change
the character of the object logic. One of the more destructive
examples of this happening in practice was when overzealous use of
unification in the Agda proof assistant~\cite{norell:2009} led to the
injectivity of type constructors being derivable, whence the principle
of the excluded middle could be refuted.

Uncritical use of unification in type theories like that of Nuprl
might even lead to inconsistency, since greater care must be taken to
negotiate the intricacies of subtyping and functionality which arise
under the semantic typing discipline. As Cockx, Devriese and Piessens
point out in their recent refit of Agda's unification theory,
unification must be integrated tightly with the object logic in order
to ensure soundness~\cite{cockx-devriese-piessens:2016}.

\subsubsection{Our solution: adopt a dependent refinement discipline}

Taking a step back, there are two more things that should make us
suspicious of the use of existential variables and unification in the
above:

\begin{enumerate}

\item We \emph{still} do not exhibit $A$ using refinement rules, but rather
  simply hope that at some point, we shall happen upon (by chance) a
  suitable witness for which we can check membership. In many cases,
  it will be possible to inhabit $B[?\mathfrak{m}]$ regardless of
  whether $?\mathfrak{m}$ has any referent, which will leave us in
  much the same position as we were in Nuprl: we must cook up some
  arbitrary $A$ independently without any help from the refinement rules.

\item It is a basic law of type-theoretic praxis that whatever structure exists
  at the propositional level should mirror a form of construction that already
  exists at the judgmental level, adding to it only those characteristics
  which distinguish truth from the evidence of an arbitrary judgment (e.g.\
  functionality, local character, fibrancy, etc.). In this case, the use of
  existential variables and unification seems to come out of nowhere, whereas
  we would expect a dependent sum at the type level to be defined in terms of
  some notion of dependent sum at the judgmental level.

\end{enumerate}

With the above in mind, we are led to try and revise the old \LCF{}
apparatus to support a \emph{dependent} refinement discipline,
relaxing the constructible subgoals property in such a way as to admit
a coding for the rule given at the beginning of this section,
$\mathfrak{\times_R}$.

\section{Survey of Related Works}

\subsection{Semantics of proof refinement in Nuprl}

The most detailed denotational semantics for tactic-based proof
refinement that we are aware of is contained in Todd Knoblock's PhD
thesis~\cite{knoblock:1988} vis-\`a-vis the Nuprl refinement
logic. Knoblock's purpose was to endow Nuprl with a tower of
metalogics, each reflecting the contents of the previous ones,
enabling internal reasoning about the proof refinement process itself;
this involved specifying semantic domains for (reflected) Nuprl
judgments, proofs and proof tactics in the Nuprl logic.

A detailed taxonomy of different forms of proof tactic was considered,
including \emph{search tactics} (analogous to \emph{valid tactics} in
\LCF{}), \emph{partial tactics} (tactics whose domain of applicability is
circumscribed a priori), and \emph{complete tactics} (partial tactics
which produce no subgoals when applied in their domain of
applicability).

Spiritually, our apparatus is most closely related to Knoblock's
contributions, in light of the purely semantical and denotational
approach which he pursued. The fact that Knoblock's semantic universe
for proof refinement was the Nuprl logic itself enabled users of the
system to prove internally the general effectiveness of a tactic on
some class of Nuprl sequents once and for all; then, such a tactic
could be applied without needing to be executed.

\subsection{Isabelle as a meta-logical framework}

Isabelle, a descendent of Cambridge LCF, is widely considered the gold
standard in tactic-based proof refinement today; at its core, it is
based on a version of Intuitionistic higher-order logic called
Isabelle/Pure, which serves as the logical framework for all other
Isabelle theories, including the famous Isabelle/HOL theory for
classical higher-order logic.

In Isabelle, tactics generally operate on a full proof state (as
opposed to the \enquote{local} style pioneered in LCF); a tactic is a
partial function from proof states to lazy sequences of proof
states. Note that the sequences here are used to accomodate
sophisticated, possibly non-deterministic search schemata. In contrast
to other members of the LCF family, the notion of validation
has been completely eschewed; this has relieved Isabelle of the need
to duplicate rules in both forward and backward inference, and
simplifies the correctness conditions for a tactic.

Isabelle does not address the issue of dependent refinement, instead
relying heavily on instantiation of schematic variables by
higher-order unification. Because HOL is Isabelle's main theory, this
is perhaps not so bad a state of affairs, since the most compelling
uses of dependent refinement arise in proof systems for dependent type
theory, in which a domain of discourse is identified with the
inhabitants of a type or proposition. In non-type-theoretic approaches
to logic, a proposition is proved using inference rules, whereas an
element of the domain of discourse is constructed according to a
grammar.

With that said, instantiation of schematic variables at higher type in
Isabelle is not always best done by unification alone, and often
requires manual guiding. A pathological case is the instantiation of
predicates shaped like $?\mathfrak{m}[?\mathfrak{n}]$, where it is
often difficult to proceed without further input~\cite[\S
4.2.2]{wenzel:2016}.

Under a dependent refinement discipline, however, the instantiation of
schematic variables ranging over higher predicates can be pursued with
the rules of the logic as a guide, in the same way that all other
objects under consideration are constructed: in the validations of
backward inference rules.

This insight, which is immanent in the higher-order
propositions-as-types principle, is especially well-adapted for use in
implementations of Martin-L\"of-style dependent type theory, where it
is often the case that the logic can guide the instantiation of a
predicate variable in ways that pure unification cannot. This is
essentially the difference between a direct and algorithmic
treatment of synthetic judgment, and its declarative simulation
as analytic judgment~\cite{martin-lof:1994}.

We stress that higher-order unification is an extremely useful
technique, but it appears to complement rather than obviate a
proper treatment of dependent refinement. Though it is not the topic
of this paper, we believe that a combination of the two techniques
would prove very fruitful.

\subsection{OLEG \& Epigram: tactics as admissible rules}

Emanating from Conor McBride's dissertation work is a unique approach
to proof elaboration which has been put to great use in several proof
assistants for dependent type theory including OLEG, Epigram and
Idris~\cite{mcbride:1999,
  mcbride:2004,brady:2005,brady:2013,christiansen-brady:2016}. A more
accessible introduction to McBride's methodology is given
in~\cite{gundry-mcbride-mckinna:2010}.

Like ours, McBride's approach rests upon the specification of
judgments in context which are stable under context substitution;
crucially, McBride's apparatus was the first treatment of proof
refinement to definitively rule out invalid scoping of schematic
variables, a problem which plagued early implementations of
\enquote{holes}. In particular, McBride's framework is well suited to
the development of type checkers, elaborators and type inferencers for
formal type theories and programming languages.

One of the ways in which our contribution differs from McBride's
system is that we treat rules of inference algebraically, i.e.\ as
first-class entities in a semantic domain together with some
operations; then, following the \LCF{} tradition, we develop a
menagerie of combinators (tacticals) by which we can combine these
rules into composite proofs. In this sense, our development is a
treatment of derivability relative to a trusted basis of
(backward) inference rules.

McBride's approach is, on the contrary, to take the trusted basis of
inference rules as given ambiently rather than algebraically,
and then to develop a theory of proof tactic based on
admissibility with respect to this basis theory.

Both approaches have been shown to be useful in implementations of type theory,
and we hope to better understand through practice the various trade-offs which
they induce.

\subsection{Typed tactics with Mtac}

The idea of capturing proof tactics using a monad is put to use in the Mtac
language, which is an extension to Coq which supports typed tactic
programming~\cite{ziliani-dreyer-krishnaswami-nanevski-vafeiadis:2013}.

\subsection{Dependent subgoals in Coq 8.5}

In his PhD thesis, Arnaud Spiwack addressed the lack of dependent
refinement in the Coq proof assistant by redesigning its tactic
apparatus~\cite{spiwack:2011}; Spiwack's efforts culminated in the
release of Coq~8.5 in early 2016, which incorporates his new design,
including support for \enquote{dependent subgoals} (which we call
\emph{dependent refinement}) and tactics that operate simultaneously
on multiple goals (which we call \enquote{multitactics}).

Spiwack's work centered around a new formulation of \LCF-style tactics
which was powerful enough to support a number of useful features,
including backtracking and subgoals expressing dependencies on other
subgoals; the latter is effected through an imperative notion of
existential variable (in contrast to the purely functional semantics
for subgoal dependency that we give in this paper).

\subsection{Our contributions}

We take a very positive view of Spiwack's contributions in this area,
especially in light of the successful concrete realization of his
ideas in the Coq proof assistant. As far as engineering is concerned,
we consider Spiwack to have definitively resolved the matter of
dependent refinement for Coq.

At the same time, we believe that there is room for a mathematical
treatment of dependent refinement which abstracts from the often
complicated details of real-world implementations, and is completely
decoupled from specific characteristics of a particular logic or proof
assistant; our experience suggests that the development of a semantics
for proof refinement along these lines can also lead to a cleaner,
more reusable concrete realization.

Our contribution is a precise, compositional and purely functional
semantics for dependent proof refinement which is also immediately
suitable for implementation; we have also introduced a novel
behavioral distinction between refinement rules and tactics based on
naturality.
Our framework is called \DependentLCF{}, and our \StandardML{}
implementation has already been used to great effect in the new
\RedPRL{} proof assistant for computational cubical type
theory~\cite{redprl:2016,angiuli-harper-wilson-popl:2016}.

\section{Preliminaries}
\subsection{Lawvere Theories}

We wish to study the algebraic structure of dependent proof refinement
for a fixed language of constructions or evidence.
To abstract away from the bureaucratic details of a particular
encoding, we will work relative to some multi-sorted Lawvere theory
$\Th$, a strictly associative cartesian category whose objects can be
viewed as sorts or contexts (finite products of sorts) and whose
morphisms may be viewed as terms or substitutions.

\begin{definition}[Lawvere theory]
  To define the notion of a multi-sorted Lawvere theory, we fix a set
  of sorts $\Sorts$; let $\Ctxs$ be the free strict associative
  cartesian category on $\Sorts$. Then, an $\Sorts$-sorted Lawvere
  theory is a strictly associative cartesian category $\Th$ equipped
  with a cartesian functor $\Of{k}{\Ctxs\to\Th}$ which is essentially
  surjective.
\end{definition}

We will write $\MultiOf{\Gamma,\Delta,\Xi}{\Th}$ for the objects of
$\Th$ and
$\MultiOf{\mathfrak{a},\mathfrak{b},\mathfrak{c}}{\ThHom{\Gamma}{\Delta}}$
for its morphisms. We will freely interchange \enquote{context} and
\enquote{sort} (and \enquote{substitution} and \enquote{term}) when
one is more clear than the other. We will sometimes write
$\IMode{\Gamma,x:\Delta}$ for the context
$\OMode{\Gamma\times\Delta}$, and then use $\IMode{x}$ elsewhere as
the canonical projection
$\OBox{\Of{\mathsf{p}}{\ThHom{\Gamma\times\Delta}{\Delta}}}$.

\begin{remark}[Second-order theories]
  In the simplest case, a Lawvere theory $\Th$ forms the category of
  contexts and substitutions for some first-order
  language. However, as Fiore and Mahmoud have shown, this machinery
  scales up perfectly well to the case of second-order theories
  (theories with binding)~\cite{fiore-mahmoud:2010}.

  In that case, the objects are contexts of second-order variables
  associated to valences (a sort together with a list of sorts
  of bound variables), and the maps are second-order substitutions;
  when the output of a map is a single valence $\vec{\sigma}.\tau$,
  the map can be read as a term binder.

  One of our reasons for specifying no more about $\Th$ than we have
  done so far is to ensure that our apparatus generalizes well to the
  case of second-order syntax, which is what is necessary in nearly
  every concrete application of this work.
\end{remark}

In what follows, we will often refer to variables as \enquote{schematic
  variables} in order to emphasize that these are variables which
range over evidence in the proof refinement apparatus, as opposed to
variables from the object logic. In the first-order case, all
variables are schematic variables; in the second-order case, the
second-order variables (called \enquote{metavariables} by Fiore~et~al)
are the schematic variables, and the object variables are essentially
invisible to our development.

\subsection{Questions Concerning a Semantic Universe}

Our main task is to define a semantic universe in which we can build
objects indexed in $\Th$, which respect substitutions of
schematic variables. Some kind of presheaf category, then, seems to be what
we want---and then proof refinement rules should be natural
transformations in this presheaf category.

The question of which indexing category to choose is a subtle one; in
order to construct our proof states monad, we will prefer to work with
something like presheaves over $\Th$, i.e.\ \emph{variable sets} which
implement all substitutions. However, most interesting refinement
rules that we wish to define will not commute with substitutions in
all cases, which is the content of naturality. This corresponds to the
fact that a refinement rule may fail to be applied if there is a
schematic variable in a certain position, but may succeed if that
variable is substituted for by some suitable term.

Essentially the same problem arises in the context of coalgebraic
logic programming~\cite{bonchi-zanasi:2013,komendantskaya-power:2016};
several methods have been developed to deal with this behavior,
including switching to an order-enriched semantic universe and using
lax natural transformations for the operational semantics; another
approach, called \enquote{saturation}, involves trivializing
naturality by treating $\Th$ as a discrete category $\Dom\Th$ (by
taking the free category on the set of objects of $\Th$), and then
saturating constructions along the adjunction
$\IsAdjoint{\IBox{\Of{i^*}{\Psh{\Th}\to\Psh{\Dom{\Th}}}}}{\IBox{\Of{i_*}{\Psh{\Dom\Th}\to\Psh{\Th}}}}$.

In the context of general dependent proof refinement, the lax
semantics are the most convenient; we will apply a variation on this
approach here, which also incorporates discrete reindexing for
interpreting tactics.

\begin{notation}
  Following the notation of the French
  school~\cite[p.~25]{maclane-moerdijk:1992}, we write
  $\IMode{\Psh{\mathbb{X}}}$ for the category of presheaves
  $\OMode{\SET^{\OpCat{\mathbb{X}}}}$ on a category $\mathbb{X}$.
\end{notation}

\begin{notation}
  We will write $\Of{\CtxsPsh}{\Psh{\Th}}$ for the constant
  presheaf of $\Th$-objects, $\Define{\CtxsPsh(\Gamma)}{\mathsf{ob}(\Th)}$.
  We may also write $\IBox{\InFib{\Gamma}{X}{F}}$ to mean
  $\OBox{\Member{X}{F(\Gamma)}}$ when $\Of{F}{\Psh{\Th}}$.
\end{notation}

We will frequently have need for a presheaf of terms of an appropriate
sort relative to a particular context,
$\Of{(\Gamma\vdash\Delta)}{\Psh{\Th}}$. This we can define informally
as follows:
\[
  \infer={
    \InFib{\Xi}{\mathfrak{a}}{(\Gamma\vdash\Delta)}
  }{
    \IMode{\Xi,\Gamma}\vdash\IMode{\mathfrak{a}}:\IMode{\Delta}
  }
\]

Formally, this is the exponential
$\OMode{\Yoneda{\Delta}^{\Yoneda{\Gamma}}}$ with
$\Of{\Yoneda{-}}{\Th\to\Psh{\Th}}$ the Yoneda embedding; this
perspective is developed in Appendix~\ref{appendix:formal-defs}.

\subsection{Presheaves and lax natural transformations}

Let $\POS$ be the order-enriched category of partially ordered sets; arrows are
endowed with an order by pointwise approximation:
$\IsApprox{f}{g}$ iff $\IBox{\IsApprox{x}{y}}\Rightarrow
\IBox{\IsApprox{f(x)}{g(y)}}$.

\begin{definition}[Presheaves and lax natural transformations]
  A \emph{$\POS$-valued presheaf} on $\mathbb{C}$ is a functor from
  $\OpCat{\mathbb{C}}$ into $\POS$.

  A \emph{lax natural transformation} $\IMode{\phi}$ between two such
  presheaves $\IMode{P},\IMode{Q}$ is a collection of components whose naturality
  square commutes up to approximation in the following sense:
  \[
    \begin{tikzcd}[sep=large]
      \IMode{P(d)}
      \arrow[r, "\IMode{P(f)}"]
      \arrow[d, swap, "\IMode{\phi_d}"]
      &
      \IMode{P(c)}
      \arrow[d, "\IMode{\phi_c}"]
      \\
      \IMode{Q(d)}
      \arrow[r, swap, "\IMode{Q(f)}"]
      \arrow[ur, phantom, "\preccurlyeq"]
      &
      \IMode{Q(c)}
    \end{tikzcd}
  \]
  In other words, we need have only that
  $\IsApprox{Q(f)\circ\phi_d}{\phi_c\circ P(f)}$ in the
  above diagram.
\end{definition}

\begin{notation}
  We will write $\LaxPsh{\mathbb{C}}$ for the category of
  $\POS$-valued presheaves and lax natural transformations on
  $\mathbb{C}$. \emph{Note that the presheaves themselves are strict;
  only natural transformations between the presheaves are lax.} A
  $\SET$-valued presheaf $\Of{P}{\Psh{\mathbb{C}}}$ can be silently
  regarded as a $\POS$-valued presheaf $\Of{P}{\LaxPsh{\mathbb{C}}}$
  by endowing each fiber with the discrete order.
\end{notation}

\section{A Framework for Proof Refinement}

We will now proceed to develop the \DependentLCF{} theory by
specifying the semantic objects under consideration, namely
\emph{judgment structures}, \emph{proof states}, \emph{refinement
  rules}, and \emph{proof tactics}.

\subsection{Judgment Structures}

A \emph{judgment} is an intention toward a particular form of
construction; that is, a form of judgment is declared by specifying
the $\Th$-object (that is, sort or context) which classifies what it
intends to construct. It is suggestive to consider this object the
\emph{output} of a judgment, in the sense that if the judgment is
derived, it will emit a substitution of the appropriate sort which can
be used elsewhere.

For example, in a Martin-L\"of-style treatment of intuitionistic logic
(see~\cite{martin-lof:1996}), the judgment $\IBox{\IsTrue{P}}$
constructs objects of sort $\OMode{\SortExp}$, where $\SortExp$ is the
sort of expressions in a basic programming language.

In the dependent proof refinement discipline, the statement of a
judgment may be interrupted by a schematic variable (e.g.\ the
judgment $\IsTrue{P(x)}$), which ranges over the evidence of some
other judgment, and may be substituted for by a term of the
appropriate sort.
This behavior captures the ubiquitous case of existential
instantiation, where we have a predicate applied to a schematic variable
which stands for an element of the domain of discourse, to be refined
in the course of verifying another judgment.

To make this precise, we can define a notion of \enquote{judgment structure} as
a collection of \enquote{judgments} which varies over contexts and
substitutions, along with an assignment of sorts to judgments: the sort
assigned to a judgment is then the sort of object that the judgment intends to
construct.

Then, a homomorphism between judgment structures would be a natural
transformation of presheaves which preserves sort assignments.
In Section~\ref{sec:refinement-rules} we will capture \emph{refinement
  rules} as homomorphisms between certain kinds of judgment structures.

\begin{definition}[Judgment structures]\label{def:judgment-structure}
  Formally, we define the category of judgment structures
  $\IMode{\JStr}$ on $\Th$ as the slice category
  $\OMode{\Slice{\LaxPsh{\Th}}{\CtxsPsh}}$; expanding definitions, an
  object $\Of{J}{\JStr}$ is a presheaf $\Of{J}{\LaxPsh{\Th}}$ together
  with an assignment $\Of{\JProj{J}}{J\to \CtxsPsh}$. Then, a map from
  $J_0$ to $J_1$ is a natural transformation that preserves
  $\JProj{}$, in the sense that the following triangle commutes:
  \[
    \begin{tikzcd}
      \IMode{J_0}
        \arrow[rr, "\IMode{\phi}"]
        \arrow[rd, swap, "\IMode{\JProj{J_0}}"]
      &&
      \IMode{J_1}
        \arrow[ld, "\IMode{\JProj{J_1}}"]
      \\
      &
      \IMode{\CtxsPsh}
    \end{tikzcd}
  \]
\end{definition}

It will usually be most clear to define a judgment structure
inductively in syntactic style, by specifying the (meta)judgment
$\IBox{\IsJdg{\Gamma}{X}{J}{\Delta}}$, pronounced \enquote{$\IMode{X}$ is a
  $\IMode{J}$-judgment in context $\IMode{\Gamma}$, constructing a substitution for
  $\OMode{\Delta}$}, which will mean $\OBox{\Member{X}{J(\Gamma)}}$ and
$\OBox{\IsEq{\JProj{J}^\Gamma(X)}{\Delta}}$.

\begin{remark}
  It is easiest to understand this judgment in the special case where
  $\IMode{\Delta}$ is a unary context $\OMode{x:\tau}$; then the judgment
  means that the construction induced by the $J$-judgment $X$ (i.e.\
  it's \enquote{output}) will be a term of sort $\tau$; in general, we
  allow multiple outputs to a judgment, which corresponds to the case
  that $\Delta$ is a context with multiple elements.
\end{remark}

When defining a judgment structure $J$, its order will be specified
using the (meta)judgment $\IBox{\ApproxJdg{\Gamma}{X}{Y}{J}}$
(pronounced \enquote{$X$ approximates $Y$ as a $J$-judgment in context
  $\Gamma$}), which presupposes both
$\IBox{\IsJdg{\Gamma}{X}{J}{\Delta}}$ and
$\IBox{\IsJdg{\Gamma}{Y}{J}{\Delta}}$; unless otherwise specified,
when defining a judgment structure, we usually assume the discrete
order.

In practice, a collection of inference rules in (meta)judgments
$\IsJdg{\Gamma}{X}{J}{\Delta}$ and $\ApproxJdg{\Gamma}{X}{Y}{J}$
should be understood as defining the \emph{least} judgment structure
closed under those rules, unless otherwise stated.

\begin{example}[Cost dynamics of basic arithmetic]
  \label{ex:example-judgment-structure}

  A simple example of a judgment structure can be given by considering
  the cost dynamics for a small language of arithmetic
  expressions~\cite[Ch.\ 7.4]{harper:2016}.

  We will fix two syntactic sorts, $\SortNum$ and $\SortExp$; $\SortNum$
  will be the sort of numerals, and $\SortExp$ will be the sort of
  arithmetic expressions; the Lawvere theory generated from these sorts
  and suitable operations that we define in
  Figure~\ref{fig:example-theory} will be called $\ThArith$. We will
  write $\IMode{\JStrArith}$ for the category of judgment structures
  over $\ThArith$, namely the slice category
  $\OMode{\Slice{\LaxPsh{\ThArith}}{\CtxsPshArith}}$.

  Then, we define a \emph{judgment structure} $\Of{\JArith}{\JStrArith}$
  for our theory by specifying the following forms of judgment:
  \begin{enumerate}

  \item $\IMode{\Eval{e}}$ means that the arithmetic expression
    $\Of{e}{\SortExp}$ can be evaluated; its evidence is the numeral
    value of $e$ and the cost $k$ of evaluating $e$ (i.e.\ the number of
    steps taken).

  \item $\IMode{\Add{m}{n}}$ means that the numerals
    $\MultiOf{m,n}{\SortNum}$ can be added; the evidence of this
    judgment is the numeral which results from their addition.

  \end{enumerate}

  The judgment structure $\JArith$ summarized above is defined
  schematically in Figure~\ref{fig:example-theory}.
\end{example}

\begin{figure*}
  \begin{gather*}
    \infer{
      \IsTmArith{\Gamma}{\overline{n}}{\SortNum}
    }{
      \Member{\overline{n}}{\Nat}
    }
    \qquad
    \infer{
      \IsTmArith{\Gamma}{\Num{n}}{\SortExp}
    }{
      \IsTmArith{\Gamma}{n}{\SortNum}
    }
    \qquad
    \infer{
      \IsTmArith{\Gamma}{e_1 + e_2}{\SortExp}
    }{
      \IsTmArith{\Gamma}{e_1}{\SortExp}
      &
      \IsTmArith{\Gamma}{e_2}{\SortExp}
    }
    \\[6pt]
    \framebox{$\Of{\JArith}{\JStr}$}
    \qquad
    \vcenter{
      \infer{
        \IsJdg{\Gamma}{\Eval{e}}{\JArith}{x_c:\SortNum,x_v:\SortNum}
      }{
        \IsTmArith{\Gamma}{e}{\SortExp}
      }
    }
    \qquad
    \vcenter{
      \infer{
        \IsJdg{\Gamma}{\Add{m}{n}}{\JArith}{\SortNum}
      }{
        \IsTmArith{\Gamma}{m}{\SortNum}
        &
        \IsTmArith{\Gamma}{n}{\SortNum}
      }
    }
  \end{gather*}
  \caption{An example theory $\ThArith$ and judgment structure
    $\JArith$.}
  \label{fig:example-theory}
\end{figure*}

We will use the above as our running example, and after we have
defined a suitable notion of \emph{refinement rule}, we will define
the appropriate rules for the judgment structure $\JArith$.

\subsection{Telescopes and Proof States}

\subsubsection{Telescopes}

An ordered sequence of judgments in which each induces a variable of
the appropriate sort which the rest of the sequence may depend on is
called a \emph{telescope}~\cite{debruijn:1991}. The notion of a
telescope will be the primary aspect in our definition of proof states
later on, where it will specify the collection of judgments which
still need to be proved.

\begin{remark}
  If \enquote{dependent refinement} were replaced with
  \enquote{independent refinement}, then the telescope data-structure
  could be replaced with lists. This design choice characterizes the
  \LCF{} family of proof refinement apparatus.
\end{remark}

\begin{figure*}
  \begin{gather*}
    \framebox{$\Of{\Tele}{\JStr\to\JStr}$}
    \qquad
    \vcenter{
      \infer{
        \IsTele{\Gamma}{\TNil}{J}{\cdot}
      }{
      }
    }
    \qquad
    \vcenter{
      \infer{
        \IsTele{\Gamma}{\TCons{x}{X}{\Psi_x}}{J}{x:\Delta,\Xi}
      }{
        \IsJdg{\Gamma}{X}{J}{\Delta}
        &
        \IsTele{\Gamma,x:\Delta}{\Psi_x}{J}{\Xi}
      }
    }\tag{Telescopes}
    \\[6pt]
    \framebox{$\Of{\State}{\JStr\to\JStr}$}
    \qquad
    \vcenter{
      \infer{
        \IsJdg{\Gamma}{(\MkStateVerbose{\Psi}{\mathfrak{a}}{\Delta})}{\State(J)}{\Delta}
      }{
        \IsTele{\Gamma}{\Psi}{J}{\Xi}
        &
        \InFib{\Gamma}{\mathfrak{a}}{(\Xi\vdash\Delta)}
      }
    }
    \qquad
    \vcenter{
      \infer{
        \IsJdg{\Gamma}{\Fail[\Delta]}{\State(J)}{\Delta}
      }{
      }
    }
    \qquad
    \vcenter{
      \infer{
        \IsJdg{\Gamma}{\bot[\Delta]}{\State(J)}{\Delta}
      }{
      }
    }
    \tag{Proof States}
    \\[6pt]
    \infer{
      \ApproxJdg{\Gamma}{\bot[\Delta]}{S}{\State(J)}
    }{
      \IsJdg{\Gamma}{S}{\State(J)}{\Delta}
    }
    \tag{Approximation}
    \\[6pt]
    \begin{aligned}
      \ADefine{\MkState{\Psi}{\mathfrak{a}}}{\MkStateVerbose{\Psi}{\mathfrak{a}}{\Delta}}
      \\
      \ADefine{\Fail}{\Fail[\Delta]}
      \\
      \ADefine{\bot}{\bot[\Delta]}
    \end{aligned}
    \tag{Notations}
  \end{gather*}
  \caption{Definitions of telescopes and proof states. Except where otherwise specified, the discrete order is assumed in all definitions above.}
  \label{fig:telescopes-and-states}
\end{figure*}

We intend telescopes to themselves be an endofunctor on judgment
structures, analogous to an iterated dependent sum; we will define the
judgment structure endofunctor $\Of{\Tele}{\JStr\to\JStr}$
inductively.

For a judgment structure $\Of{J}{\JStr}$, a $J$-telescope is either
$\TNil$ (the empty telescope), or $\TCons{x}{X}{\Psi}$ with $X$ a
$J$-judgment and $\Psi$ a $J$-telescope with $x$ bound; the sort that
is synthesized by a telescope is the product of the sorts synthesized
by its constituent judgments. The precise rules for forming telescopes
are given in Figure~\ref{fig:telescopes-and-states}.

Note how in the above, we have used variable binding notation;
formally, as can be seen in Appendix~\ref{appendix:formal-defs}, this
corresponds to exponentiation by a representable functor, which is the
usual way to account for variable binding in higher algebra. To be
precise, given a presheaf $\Of{P}{\Psh{\Th}}$ and a variable context
$\Of{\Gamma}{\Th}$, by exponentiation we can construct a new presheaf
$\Of{P^{\Yoneda{\Gamma}}}{\Psh{\Th}}$ (with
$\Of{\Yoneda{-}}{\Th\to\Psh\Th}$ the Yoneda embedding) whose values
are binders that close over the variables in $\Gamma$.

Henceforth, we are justified in adopting the \emph{variable
  convention}, by which terms are identified up to renamings of their
bound variables.

\subsubsection{Proof States Monad}

We define another endofunctor on judgment structures for proof states,
$\Of{\State}{\JStr\to\JStr}$. Fixing a judgment structure
$\Of{J}{\JStr}$, there are three ways to form a $J$-proof state:
\begin{enumerate}
\item When $\Psi$ is a $J$-telescope that synthesizes sorts $\Xi$, and
  $\mathfrak{a}$ is a substitution from $\Xi$ to $\Delta$, then
  $\MkStateVerbose{\Psi}{\mathfrak{a}}{\Delta}$ is a proof state; as
  an object in a judgment structure, this proof state synthesizes
  $\Delta$. Intuitively, $\Psi$ is the collection of subgoals and
  $\mathfrak{a}$ is the \emph{validation} which constructs evidence on
  the basis of the evidence for those subgoals. We will usually write
  $\MkState{\Psi}{\mathfrak{a}}$ instead of
  $\MkStateVerbose{\Psi}{\mathfrak{a}}{\Delta}$ when it is clear from
  context.

\item For any $\Delta$, $\Fail[\Delta]$ is a proof state that
  synthesizes $\Delta$; this state represents \emph{persistent
    failure}. We will always write $\Fail$ instead of $\Fail[\Delta]$.

\item For any $\Delta$, $\bot[\Delta]$ is a proof state that
  synthesizes $\Delta$; this state represents what we call
  \emph{unsuccess}, or failure which may not be persistent. As
  above, we will always write $\bot$ instead of $\bot[\Delta]$.
\end{enumerate}

Moreover, we impose the approximation ordering
$\ApproxJdg{\Gamma}{\bot[\Delta]}{\MkStateVerbose{\Psi}{\mathfrak{a}}{\Delta}}{\State(J)}$
and $\ApproxJdg{\Gamma}{\bot[\Delta]}{\Fail[\Delta]}{\State(J)}$.

The difference between \emph{unsuccess} and \emph{failure} is closely
related to the difference between the statements \enquote{It is not
  the case that $P$ is true} and \enquote{$P$ is false} in
constructive mathematics. In the context of proof refinement in the
presence of schematic variables, it may be the case that a rule does
not apply at first, but following a substitution, it does apply;
capturing this case is the purpose of introducing the $\bot$ proof
state.

Again, the precise rules for forming proof states are given in
Figure~\ref{fig:telescopes-and-states}.

\begin{notation}
  We will write $\TConcat{\Psi}{\Psi'}$ for the concatenation of two
  telescopes, where $\Psi'$ may have variables from
  $\JProj{\Tele(J)}^\Gamma(\Psi)$ free.  Likewise, we will write
  $\IMode{\WkSt{\Psi'}{S}}$ to mean
  $\OMode{\MkState{\TConcat{\Psi'}{\Psi}}{\mathfrak{a}}}$ when
  $\Match{S}{\MkState{\Psi}{\mathfrak{a}}}$, and $\OMode{\Fail}$ when
  $\Match{S}{\Fail}$, and $\OMode{\bot}$ when $\Match{S}{\bot}$.
\end{notation}

We can instantiate a monad structure on proof states, which will abstractly
implement the composition of refinement rules, and which will play a crucial
part in the identity and sequencing tacticals from \LCF.
\[
  \begin{tikzcd}
    \IMode{\ArrId{}}
    \arrow[r, "\IMode{\eta}"]
    &
    \IMode{\State}
    &
    \IMode{\State\circ\State}
    \arrow[l, swap, "\IMode{\mu}"]
  \end{tikzcd}
\]
The unit and multiplication operators are defined by the following equations:
\begin{align*}
  \AIsEq{\StUnit{\Gamma}{X}}{
    \OMode{
      \MkState{\TCons{x}{X}{\TNil}}{x}
    }
  }
  \\
  \AIsEq{\StMul{\Gamma}{\MkState{\TNil}{\mathfrak{a}}}}{
    \OMode{
      \MkState{\TNil}{\mathfrak{a}}
    }
  }
  \\
  \AIsEq{\StMul{\Gamma}{\MkState{\TCons{x}{(\MkState{\Psi_x}{\mathfrak{a}_x})}{\Ul{\Psi}}}{\mathfrak{a}}}}{
    \OMode{
      \WkSt{\Psi_x}{
        \StMul{\Gamma,\JProj{\Tele(J)}^\Gamma(\Psi_x)}{
          \MkState{\Ul{\Psi}}{\mathfrak{a}}
        }[\mathfrak{a}_x/x]
      }
    }
  }
  \\
  \AIsEq{\StMul{\Gamma}{\MkState{\TCons{x}{\Fail}{\Ul{\Psi}}}{\mathfrak{a}}}}{
    \OMode{\Fail}
  }
  \\
  \AIsEq{\StMul{\Gamma}{\MkState{\TCons{x}{\bot}{\Ul{\Psi}}}{\mathfrak{a}}}}{
    \OMode{\bot}
  }
  \\
  \AIsEq{\StMul{\Gamma}{\Fail}}{\OMode{\Fail}}
  \\
  \AIsEq{\StMul{\Gamma}{\bot}}{\OMode{\bot}}
\end{align*}
In the interest of clear notation, we have used $\Psi$ to range over
$\Tele(J)$ and $\Ul{\Psi}$ to range over $\Tele(\State(J))$.

\begin{restatable}{thm}{CoherenceLemma}\label{thm:coherence}
  Proof states form a monad on $\JStr$, i.e.\ the following diagrams
  commute:
  \[
    \begin{tikzcd}[sep=huge, ampersand replacement = \&]
      \IMode{\State}
      \arrow[r, "\IMode{\eta}"]
      \arrow[rd, swap,"\IMode{\ArrId{}}"]
      \&
      \IMode{\State\circ\State}
      \arrow[d, "\IMode{\mu}"]
      \&
      \IMode{\State}
      \arrow[l, swap, "\IMode{\State(\eta)}"]
      \arrow[ld, "\IMode{\ArrId{}}"]
      \\
      \&
      \IMode{\State}
    \end{tikzcd}
  \]
  \[
    \begin{tikzcd}[sep=huge, ampersand replacement = \&]
      \IMode{\State\circ\State\circ\State}
      \arrow[r, "\IMode{\mu}"]
      \arrow[d, swap, "\IMode{\State(\mu)}"]
      \&
      \IMode{\State\circ\State}
      \arrow[d,"\IMode{\mu}"]
      \\
      \IMode{\State\circ\State}
      \arrow[r,"\IMode{\mu}"]
      \&
      \IMode{\State}
    \end{tikzcd}
  \]
\end{restatable}
\begin{proof}
  By nested induction; see Appendix~\ref{appendix:proofs}.
\end{proof}

\subsection{Refinement Rules and Lax Naturality}\label{sec:refinement-rules}

We can now directly define the notions of \emph{refinement rule},
\emph{tactic} and \emph{multitactic} in terms of judgment structure
homomorphisms.

\begin{definition}[Refinement rules]\label{def:refinement-rule}
  For judgment structures $\MultiOf{J_0,J_1}{\JStr}$, a
  \emph{refinement rule} from $J_0$ to $J_1$ is a $\JStr$-homomorphism
  $\Of{\rho}{\IBox{\Define{\Rules{J_0}{J_1}}{J_0\to\State(J_1)}}}$. Unpacking
  definitions, $\rho$ is a \emph{lax} natural transformation between
  the underlying presheaves of $J$ and $\State(J_1)$ which preserves
  the projection $\JProj{}$.

  Usually, one works with homogeneous refinement rules
  $\Of{\rho}{J\to\State(J)}$, which can be called $J$-rules.
\end{definition}

The ordered character of the $\State(J)$ judgment structure is crucial
in combination with lax naturality; it is this which allows us to
define a refinement rule which neither succeeds nor fails when it
encounters a schematic variable that is blocking its applicability:
that is, it does not commit to failure under all possible
instantiations for that variable.

Full naturality would entail that a refinement rule commute with all
substitutions from $\Th$, whereas lax naturality only requires this
square to commute up to approximation---that is, for
$\Of{P,Q}{\LaxPsh{\Th}}$, $\Of{\phi}{P\to Q}$ and
$\Of{\mathfrak{a}}{\ThHom{\Gamma}{\Delta}}$, rather than the identity
$\IsEq{Q(\mathfrak{a})\circ\phi_\Delta}{\phi_\Gamma\circ
  P(\mathfrak{a})}$, we require only the approximation
$\IsApprox{Q(\mathfrak{a})\circ\phi_\Delta}{\phi_\Gamma\circ
  P(\mathfrak{a})}$.

To understand why this is desirable, let us return to the example of
the judgment form $\IsTrue{P}$; supposing that our ambient theory
$\Th$ has a sort for propositions $\Of{\SortProp}{\Th}$ and a sort for
program expressions $\Of{\SortExp}{\Th}$, we can define a judgment
structure $\Of{L}{\JStr}$ by specifying a single form of judgment:
\[
  \infer{
    \IsJdg{\Gamma}{\IBox{\IsTrue{P}}}{L}{\SortExp}
  }{
    \IsTm{\Gamma}{P}{\SortProp}
  }
\]
Then, our task is to code the inference rules of our logic, such as
the following,
\[
  \infer[\lor I_1]{
    \IsTrue{P\lor Q}
  }{
    \IsTrue{P}
  }
\]
as \emph{refinement rules} for the judgment structure $L$. Such a rule
is a judgment structure homomorphism $\Of{\land I}{L\to\State(L)}$ in
$\JStr$; at first, we might try and write the following:
\begin{align*}
  \IMode{\lor I_1^\Gamma}
  &
  \begin{RuleDefn}
    \RuleCase{\IBox{\IsTrue{P\lor Q}}}{
      \MkState{\TCons{x}{\OBox{\IsTrue{P}}}{\TNil}}{
        \mathsf{inl}(x)
      }
    }
    \\
    \RuleCase{\_}{\Fail}
  \end{RuleDefn}
  \tag{*}
\end{align*}

This, however, is not a well-formed definition, because it does not
commute in any sense with substitutions: for instance
$\IsEq{\lor I_1^{\Gamma,x:\SortProp}(\IBox{\IsTrue{x}})[P\lor
  Q/x]}{\Fail}$, whereas
$\IsEq{\lor I_1^{\Gamma}(\IBox{\IsTrue{x}}[P\lor
  Q/x])}{\MkState{\TCons{x}{\IBox{\IsTrue{P}}}{\TNil}}{\mathsf{inl}(x)}}$. However,
recall that refinement rules are subject only to lax naturality, i.e.\
naturality up to approximation; with a small adjustment to our
definition, we can make it commute with substitutions up to
approximation:
\begin{align*}
  \IMode{\lor I_1^\Gamma}
  &
  \begin{RuleDefn}
    \RuleCase{\IBox{\IsTrue{P\lor Q}}}{
      \MkState{\TCons{x}{\OBox{\IsTrue{P}}}{\TNil}}{
        \mathsf{inl}(x)
      }
    }
    \\
    \RuleCase{x}{\bot}
    \\
    \RuleCase{\_}{\Fail}
  \end{RuleDefn}
\end{align*}

Indeed, the above definition is well-formed (because we have
$\ApproxJdg{\Gamma}{\bot}{\MkState{\TCons{x}{\IBox{\IsTrue{P}}}{\TNil}}{\mathsf{inl}(x)}}{\State(L)}$). This
example reflects the difference between \emph{unsuccess} and
\emph{failure}: namely, the introduction rule above does not
\emph{yet} apply to the goal $\IsTrue{x}$, but supposing $x$ were
substituted for by some $P\lor Q$, it would then apply. On the other
hand, the rule does not apply at all to the goal $\IsTrue{P\land Q}$.

\begin{example}[Refinement rules for cost dynamics]\label{ex:example-refinement-rules}
  It will be instructive to consider a more sophisticated example.
  Resuming what we started in
  Example~\ref{ex:example-judgment-structure}, we are now equipped to
  encode formal refinement rules for the judgment structure
  $\Of{\JArith}{\JStr}$ defined in Figure~\ref{fig:example-theory}.

  We will implement the cost dynamics using two evaluation rules and one
  rule to implement the addition of numerals:
  \begin{align*}
    \AOf{
    \RuleNumEval
    }{
    \JArith\to\State(\JArith)
    }
    \\
    \AOf{
    \RulePlusEval
    }{
    \JArith\to\State(\JArith)
    }
    \\
    \AOf{
    \RuleAdd
    }{
    \JArith\to\State(\JArith)
    }
  \end{align*}

  First, let's consider what these rules would look like informally on
  paper, writing $\PrettyEval{e}{k}{n}$ for the statement that the
  judgment $\IMode{\Eval{e}}$ obtains, synthesizing cost $\OMode{k}$ and
  numeral $\OMode{n}$, and writing $\Match{m+n}{o}$ for the statement
  that the judgment $\IMode{\Add{m}{n}}$ obtains, synthesizing numeral
  $\OMode{o}$:
  \[
    \infer[\RuleNumEval]{
      \PrettyEval{\Num{n}}{0}{n}
    }{
    }
  \]
  \[
    \infer[\RulePlusEval]{
      \PrettyEval{e_1+e_2}{k}{n}
    }{
      \begin{array}{l}
        \PrettyEval{e_1}{k_1}{n_1}\\
        \PrettyEval{e_2}{k_2}{n_2}
      \end{array}
      \begin{array}{l}
        \Match{k_1+k_2}{k_{12}}\\
        \Match{\overline{1}+k_{12}}{k}
      \end{array}
      &
      \Match{n_1+n_2}{n}
    }
  \]
  \[
    \infer[\RuleAdd]{
      \Match{\overline{m}+\overline{n}}{\overline{m+n}}
    }{
    }
  \]

  In keeping with standard practice and notation, in the informal
  definition of a refinement rule, clauses for failure and unsuccess are
  elided. When we code these rules as judgment structure homomorphisms,
  we add these clauses in the appropriate places, as can be seen from
  the formal definitions of $\RuleNumEval$, $\RulePlusEval$ and
  $\RuleAdd$ in Figure~\ref{fig:example-refinement-rules}.

\end{example}

\begin{figure*}
  \begin{align*}
    \IMode{\RuleNumEval^\Gamma}
    &
    \begin{RuleDefn}
      \RuleCase{\Eval{\Num{m}}}{
        \MkState{\TNil}{[\overline{0},m]}
      }
      \\
      \RuleCase{\Eval{x}}{\bot}
      \\
      \RuleCase{\_}{\Fail}
    \end{RuleDefn}
    \\
    \IMode{\RulePlusEval^\Gamma}
    &
    \begin{RuleDefn}
      \RuleCase{\Eval{e_1+e_2}}{
        \MkState{
          \begin{array}[c][{l}]
            {[x_c,x_v]}:\Eval{e_1}.\
            {[y_c,y_v]}:\Eval{e_2}.\\
            z_c:\Add{x_c}{y_c}.\
            z_c':\Add{\overline{1}}{z_c}.\
            z_v:\Add{x_v}{y_v}
          \end{array}
        }{[z_c',z_v]}
      }
      \\
      \RuleCase{\Eval{x}}{\bot}
      \\
      \RuleCase{\_}{\Fail}
    \end{RuleDefn}
    \\
    \IMode{\RuleAdd_\Gamma}
    &
    \begin{RuleDefn}
      \RuleCase{\Add{\overline{m}}{\overline{n}}}{
        \MkState{\TNil}{\overline{m+n}}
      }
      \\
      \RuleCase{\Add{\_}{\_}}{\bot}
      \\
      \RuleCase{\_}{\Fail}
    \end{RuleDefn}
  \end{align*}
  \caption{Defining refinement rules for $\JArith$, the judgment
    structure of cost dynamics for arithmetic expressions.}
  \label{fig:example-refinement-rules}
\end{figure*}

\subsection{Combinators for Derived Refinement Rules}\label{sec:rule-combinators}

We can develop a menagerie of combinators for refinement rules which allow the
development of \emph{derived} rules. More generally, given a basis of
refinement rules, it is possible to characterize the space of \emph{derivable
rules} by the closure of this basis under certain combinators (see
Appendix~\ref{appendix:defining-refinement-logics}).

First, we will begin by defining an auxiliary judgment structure
$\Of{\Labeled{J}}{\JStr}$, which tags $J$-judgments with an index:
\[
  \framebox{$\Of{\Labeled{J}}{\JStr}$}
  \qquad
  \vcenter{
    \infer={
      \IsJdg{\Gamma}{\Tuple{X,i}}{\Labeled{J}}{\Delta}
    }{
      \IsJdg{\Gamma}{X}{J}{\Delta}
      &
      \Member{i}{\Nat}
    }
  }
\]

Next, define an operation to label the subgoals of a proof state with their index:
\begin{align*}
  \IMode{\LABEL} &:\IMode{\State(J)\to\State(\Labeled{J})}
  \\
  \IMode{\LABEL_\Gamma}
  &
  \;\begin{RuleDefn}
    \RuleCase{\MkState{\Psi}{\mathfrak{a}}}{
      \MkState{\LABEL_0^\Gamma(\Psi)}{\mathfrak{a}}
    }
    \\
    \RuleCase{\bot}{\bot}
    \\
    \RuleCase{\Fail}{\Fail}
  \end{RuleDefn}
\end{align*}
where
\begin{align*}
  \IMode{\LABEL_i} &:\IMode{\Tele(J)\to\Tele(\Labeled{J})}
  \\
  \IMode{\LABEL_i^\Gamma}
  &
  \;\begin{RuleDefn}
    \RuleCase{\TNil}{\TNil}
    \\
    \RuleCase{\TCons{x}{X}{\Psi}}{
      \TCons{x}{\Tuple{X,i}}{\LABEL_{i+1}^{\Gamma,x:\JProj{J}^\Gamma(X)}(\Psi)}
    }
  \end{RuleDefn}
\end{align*}

Let $\IMode{\vec{\rho}}$ range over a list of rules $\Of{\rho_i}{\Rules{J_0}{J_1}}$;
now we define a derived rule which applies the appropriate rule to a labeled
judgment:
\begin{align*}
  \AOf{\PROJ(\vec{\rho})}{\Rules{\Labeled{J}}{J}}
  \\
  \ADefine{
    {\PROJ(\vec{\rho})}_\Gamma\Tuple{X,i}
  }{
    \begin{cases}
        \rho_i^\Gamma(X) &{\normalcolor\textbf{if } \IMode{i}<\IMode{\vert\vec{\rho}\vert}}
        \\
        \eta_\Gamma(X) &{\normalcolor\textbf{otherwise}}
    \end{cases}
  }
\end{align*}

Now, given rules $\Of{\rho}{\Rules{J_0}{J_1}}$ and
$\Of{\vec{\rho}}{\List{\Rules{J_1}{J_2}}}$, we can define the composition
$\Of{\rho;\vec{\rho}}{\Rules{J_0}{J_2}}$ using the multiplication operator of
the $\State$ monad as follows:
\begin{align*}
  \IMode{\rho;\vec{\rho}} &: \IMode{\Rules{J_0}{J_2}}
  \\
  \ADefine{
    \rho;\vec{\rho}
  }{
    \mu
    \circ
    \State(\PROJ(\vec{\rho}))
    \circ
    \LABEL
    \circ
    \rho
  }
\end{align*}

% Now we can make precise the notion of \emph{derivability} relative to a
% collection of primitive refinement rules.
% \begin{definition}[Derivability closure]
%   For a logical theory on a judgment structure $\Of{J}{\JStr}$ defined by
%   a collection of primitive inference rules
%   $\IsSubsetEq{\Sigma}{\Rules{J}{J}}$, the \emph{derivability closure}
%   $\Sigma^\star$ of $\Sigma$ is given by the following definition:
%   \begin{gather*}
%     \infer{
%       \Member{\rho}{\Sigma^\star}
%     }{
%       \Member{\rho}{\Sigma}
%     }
%     \qquad
%     \infer{
%       \Member{\rho;\vec{\rho}}{\Sigma^\star}
%     }{
%       \Member{\rho}{\Sigma^\star}
%       &
%       \Member{\rho_i}{\Sigma^\star}\ (\IMode{i} < \IMode{\Dom{\vec{\rho}}})
%     }
%   \end{gather*}
%   In other words, the \emph{derivable rules} relative to $\Sigma$ are the
%   elements of $\Sigma^*$.
% \end{definition}
In Appendix~\ref{appendix:defining-refinement-logics}, we develop a very
sophisticated fibered categorical notion of closed refinement logic
(\enquote{refiner}) together with a characterization of derivability relative to
refiners.

\subsection{From Refinement Rules to Tactics}

Tactics are distinguished from refinement rules in that they are not subject to
the lax naturality condition; this is because in tactic-based proof refinement,
it is necessary to support tactics which do not commute with substitutions, not
even up to approximation---for example, $\ORELSE$ and $\TRY$ will result in
completely different proof states before and after a substitution.
For this reason, we define a new category $\DiscJStr$ of discrete
judgment structures, for which tactics will be certain homomorphisms.

\begin{definition}[Discrete Judgment Structures]
  We now define the category of \emph{discrete judgment structures},
  $\Define{\DiscJStr}{\Slice{\Psh{\Dom{\Th}}}{\CtxsPsh}}$, where
  $\Dom{\Th}$ is the subcategory of $\Th$ which contains only identity
  arrows, and
  $\Define{\Psh{\Dom\Th}}{\OpCat{\Dom{\Th}}\to\SET}$.

  Then $\DiscJStr$ homomorphisms are the same as those of $\JStr$,
  except their naturality condition is trivially satisfied for any
  collection of components; this is because only identity maps occur
  $\Dom\Th$, so the naturality squares are degenerate. Furthermore,
  because the components of such homomorphisms are defined in $\SET$,
  there is no requirement of monotonicity.
\end{definition}

By composing with the inclusion $\Of{i}{\Dom{\Th}\to\Th}$ and the forgetful
functor $\Of{U}{\POS\to\SET}$, any judgment structure $\Of{J}{\JStr}$ can be
reindexed to a discrete judgment structure
$\Define{\IBox{\Of{\Dom{J}}{\DiscJStr}}}{U\circ J\circ i}$, with
$\Define{\JProj{\Dom{J}}^\Gamma}{\JProj{J}^\Gamma}$:
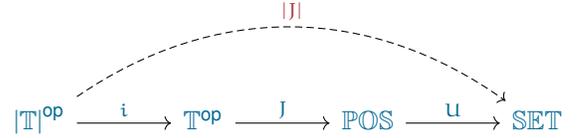
\[
  \begin{tikzcd}[sep=large]
    \IMode{\OpCat{\Dom{\Th}}}
    \arrow[r,"\IMode{i}"]
    \arrow[rrr,dashed,bend left=35,"\OMode{\Dom{J}}"]
    &
    \IMode{\OpCat{\Th}}
    \arrow[r,"\IMode{J}"]
    &
    \IMode{\POS}
    \arrow[r,"\IMode{U}"]
    &
    \IMode{\SET}
  \end{tikzcd}
\]

\begin{remark}

  If we intended to develop a theory of tactics which do not commute
  with substitution, why did we bother with the presheaf apparatus in
  the first place? In essence, the reason is that we are building a
  semantics which accounts for both refinement rules and
  tactics, and refinement rules are distinguished from other
  proof refinement strategies precisely by the characteristic of lax
  naturality.

  The purpose of tactics, on the other hand, is to subvert naturality in order
  to formalize modular \enquote{proof sketches} which are applicable to a broad class
  of goals. Indeed, this subversion of naturality by tactics is simultaneously
  the source of their unparalleled practicality as well as the cause of their
  notoriously brittle character.
  The uniformity of action induced by lax naturality lies in stark
  opposition to the modularity required of tactics; we will neutralize
  this contradiction by passing through discretization to $\DiscJStr$.

\end{remark}

\subsection{Tactics and Recursion}

In keeping with standard usage, a \emph{proof tactic} is a potentially
diverging program that computes a proof on the basis of some
collection of refinement rules. In order to define tactics
precisely, we will first have to specify how we intend to interpret
recursion.

A very lightweight way to interpret recursion is suggested by
Capretta's delay monad~\cite{capretta:2005} (the completely
iterative monad on the identity functor), a coinductive representation
of a process which may eventually return a value. We can define
a variation on Capretta's construction as a monad
$\Of{\infty}{\DiscJStr\to\DiscJStr}$ on discrete judgment structures,
defined as the greatest judgment structure closed under the
rules in Figure~\ref{fig:delay-monad}.

\begin{figure*}
  \begin{gather*}
    \framebox{$\Of{\infty}{\DiscJStr\to\DiscJStr}$}
    \qquad
    \vcenter{
      \infer{
        \IsJdg{\Gamma}{\Now{X}}{\infty J}{\Delta}
      }{
        \IsJdg{\Gamma}{X}{J}{\Delta}
      }
    }
    \qquad
    \vcenter{
      \infer{
        \IsJdg{\Gamma}{\LaterVerbose{\Delta}{X}}{\infty J}{\Delta}
      }{
        \IsJdg{\Gamma}{X}{\infty J}{\Delta}
      }
    }
    \tag{Delay Monad}
    \\[6pt]
    \Define{
      \Later{X}
    }{
      \LaterVerbose{\Delta}{X}
    }
    \tag{Notation}
  \end{gather*}
  \caption{Capretta's delay monad on discrete judgment structures.}
  \label{fig:delay-monad}
\end{figure*}

To summarize, for a judgment structure $\Of{J}{\DiscJStr}$, there are two
ways to construct a $\infty J$-judgment:
\begin{enumerate}
\item $\Now{X}$ is an $\infty J$-judgment when $X$ is a $J$-judgment.
\item $\Later{X}$ is an $\infty J$-judgment when $X$ is an $\infty J$-judgment.
\end{enumerate}

\begin{lem}[Delay monad]
  $\Of{\infty}{\DiscJStr\to\DiscJStr}$ forms a monad on
  $\DiscJStr$.
\end{lem}

We can repeat the same construction as above to acquire a monad
$\Of{\infty}{\JStr\to\JStr}$, which adds to the previous construction
the appropriate action for substitutions. Abusing notation, we will
use the same symbol for both monads when it is clear from context what
is meant; this is justified in practice, because the assignment of
objects is the same for the two monads.

\begin{notation}
  We will write $\Of{\eta_\infty}{\ArrId{\DiscJStr}\to\infty}$ and
  $\Of{\mu_\infty}{\infty\circ\infty\to\infty}$ for the unit and
  multiplication operators respectively. We will also employ the
  following notational convention, inspired by the
  \enquote{\texttt{do}-notation} used in the Haskell programming
  language for monads:
  \[
    \Define{
      x\leftarrow M; N(x)
    }{
      \mu_\infty(\infty (x\mapsto N(x))(M))
    }
  \]
\end{notation}

\begin{definition}[Tactics and multitactics]\label{def:tactic}

  A \emph{tactic} for judgment structures $\MultiOf{J_0,J_1}{\JStr}$ is a
  $\DiscJStr$-homomorphism:
  \[
    \Define{\Tactics{J_0}{J_1}}{
      \Dom{J_0}\to\infty\Dom{\State(J_1)}
    }
  \]

  Usually one works with homogeneous tactics $\Of{\phi}{\Tactics{J}{J}}$, which
  are called $J$-tactics.
  A $J$-\emph{multitactic} is a tactic for the judgment structure $\State(J)$.
\end{definition}

\subsection{Tacticals as Tactic Combinators}
\label{sec:tacticals}

At this point we are equipped to begin defining a collection of
standard \enquote{tacticals}, or tactic combinators.

% Note that the order enrichment which we have imposed on judgment structures is
% only for the purpose of enabling a purely momentary version of failure
% (\enquote{unsuccess}) which is suitable for refinement rules in the presence of
% schematic variables; tactics, on the contrary, are required to be neither
% natural with respect to substitutions in $\Th$ nor monotone in their components
% (this is a consequence of how we have defined $\Of{\DiscJStr}{\CAT}$).

\subsubsection{Tactics from Rules}

Every rule $\Of{\rho}{J_0\to\State(J_1)}$ can be made into a tactic
$\Define{\IBox{\Of{\Now{\rho}}{\Tactics{J_0}{J_1}}}}{\eta_\infty\circ\Dom{\rho}}$.

\subsubsection{Conditional Tacticals}

To begin with, we can define the join of two tactics
$\MultiOf{\phi,\psi}{\Tactics{J_0}{J_1}}$, which implements
$\ORELSE$ from \LCF:
\begin{align*}
  \AOf{\phi\oplus\psi}{\Tactics{J_0}{J_1}}
  \\
  \ADefine{
    {(\phi\oplus\psi)}_\Gamma(X)
  }{
    \begin{array}[t]{l}
      S \leftarrow \phi_\Gamma(X);
      \\
      \begin{cases}
        \Now{S} &{\normalcolor\textbf{if } \IsEq{S}{\MkState{\Psi}{\mathfrak{a}}}}
        \\
        \psi_\Gamma(X) &{\normalcolor\textbf{otherwise}}
      \end{cases}
    \end{array}
  }
\end{align*}

In $\phi\oplus\psi$ we have an example of a natural transformation which does
\emph{not} commute with substitutions; this is fine, because tactics are
defined as discrete judgment structure homomorphisms, and are therefore
subject to only trivial naturality conditions.

In combination with the identity tactic
$\Define{\IBox{\Of{\ID}{\Tactics{J}{J}}}}{\Now{\eta}}$, we can define
the $\TRY$ tactical which replaces a failure or unsuccess with the identity:
\begin{align*}
  \AOf{
    \TRY(\phi)
  }{
    \Tactics{J}{J}
  }
  \\
  \ADefine{
    \TRY(\phi)
  }{
    \phi\oplus\ID
  }
\end{align*}

\subsubsection{Multitacticals}

We will factor the \LCF{} sequencing tacticals $\THEN$ and $\THENL$ into
a combination of a \enquote{multitactical} (a tactic that operates on
$\State(J)$ instead of $J$) and a generic sequencing operation.

These multitacticals will be factored through tacticals that are
sensitive to the position of a goal within a proof state,
namely $\CONST$ and $\PROJ$.

Let $\IMode{\phi}$ range over tactics $\IMode{\Tactics{J}{J}}$, and let
$\IMode{\vec{\phi}}$ range over a list of such tactics. We will now define some
further tactics which work over labeled judgments (following
Section~\ref{sec:rule-combinators}):
\begin{align*}
  \AOf{\CONST(\phi)}{\Tactics{\Labeled{J}}{J}}
  \\
  \ADefine{
    {\CONST(\phi)}_\Gamma\Tuple{X,i}
  }{
    \phi_\Gamma(X)
  }
\end{align*}
\begin{align*}
  \AOf{\PROJ(\vec{\phi})}{\Tactics{\Labeled{J}}{J}}
  \\
  \ADefine{
    {\PROJ(\vec{\phi})}_\Gamma\Tuple{X,i}
  }{
    \begin{cases}
        \phi_i^\Gamma(X) &{\normalcolor\textbf{if } \IMode{i}<\IMode{\vert\vec{\phi}\vert}}
        \\
        \Now{\eta_\Gamma(X)} &{\normalcolor\textbf{otherwise}}
    \end{cases}
  }
\end{align*}

Now, we need to show how to turn transform an operation on labeled
judgments into a multitactic. We will need an operation to turn
a telescope of potentially diverging computations into a potentially
diverging computation of a telescope:
\begin{align*}
  \AOf{
    \AwaitTele
  }{
    \Tele(\infty J)\to \infty\Tele(J)
  }
  \\
  \IMode{\AwaitTele_\Gamma}
  &
  \;\begin{RuleDefn}
    \RuleCase{\TNil}{\Now{\TNil}}
    \\
    \RuleCase{\TCons{x}{X^\infty}{\Psi^\infty}}{
      \!\!\begin{array}[t]{l}
        X\leftarrow X^\infty;\\
        \Psi \leftarrow \AwaitTele_{\Gamma,x:\JProj{J}^\Gamma(X)}(\Psi^\infty);\\
        \Now{\TCons{x}{X}{\Psi}}
      \end{array}
    }
  \end{RuleDefn}
\end{align*}

Using the above, we can transform $\Of{\chi}{\Tactics{\Labeled{J}}{J}}$ into a
multitactic $\Of{\StApply{\chi}}{\Tactics{\State(J)}{\State(J)}}$:
\begin{align*}
  \AOf{
    \StApply{\chi}
  }{
    \Tactics{\State(J)}{\State(J)}
  }
  \\
  \IMode{{\StApply{\chi}}_\Gamma}
  &
  \;\begin{RuleDefn}
    \RuleCase{\bot}{\Now{\bot}}
    \\
    \RuleCase{\Fail}{\Now{\Fail}}
    \\
    \RuleCase{\MkState{\Psi}{\mathfrak{a}}}{
      \!\!\begin{array}[t]{l}
        \Psi'\leftarrow \AwaitTele_\Gamma({\Tele_\Gamma(\chi)}(\LABEL_0^\Gamma(\Psi)));
        \\
        \Now{\MkState{\Psi'}{\mathfrak{a}}}
      \end{array}
    }
  \end{RuleDefn}
\end{align*}
defining the auxiliary function $\Tele_\Gamma$ as
follows:\footnote{Note that this does not follow immediately from the
  functoriality of $\Of{\Tele}{\JStr\to\JStr}$, because $\chi$ is only a
  map in $\DiscJStr$.}
\begin{align*}
  \AOf{\Tele_\Gamma}{\Tactics{J_0}{J_1}\times\Tele(J_0)(\Gamma)\to\Tele(\infty\State(J_1))(\Gamma)}
  \\
  \IMode{\Tele_\Gamma(\chi)}
  &
  \;\begin{RuleDefn}
    \RuleCase{\TNil}{\TNil}
    \\
    \RuleCase{\TCons{x}{X}{\Psi}}{
      \TCons{x}{\chi_\Gamma(X)}{
        \Tele_{\Gamma,\JProj{J}^\Gamma(X)}(\chi, \Psi)
      }
    }
  \end{RuleDefn}
\end{align*}

Finally, we can define two multitacticals: $\ALL$ which applies a single tactic
to all goals, and $\EACH$ which applies a list of tactics pointwise to the
subgoals:
\begin{align*}
  \AOf{\ALL,\EACH}{\Tactics{J}{J}\to\Tactics{\State{J}}{\State{J}}}
  \\
  \ADefine{\ALL(\phi)}{\StApply{\CONST(\phi)}}
  \\
  \ADefine{\EACH(\phi)}{\StApply{\PROJ(\phi)}}
\end{align*}

\subsubsection{Generic Sequencing}

Fixing a tactic $\Of{\phi}{\Tactics{J_0}{J_1}}$ and
$\Of{\Ul{\psi}}{\Tactics{\State(J_1)}{\State(J_2)}}$
as above, we can define the sequencing of $\Ul{\psi}$ after
$\phi$ as the following composite, also displayed in Figure~\ref{fig:seq}:
\begin{align*}
  \AOf{
    \SEQ(\phi,\Ul{\psi})
  }{
    \Tactics{J_0}{J_2}
  }
  \\
  \ADefine{
    \SEQ(\phi,\Ul{\psi})
  }{
   \infty\Dom{\mu}\circ \mu_\infty\circ\infty(\Ul{\psi})\circ\phi
  }
\end{align*}
\begin{figure*}
  \[
    \begin{tikzcd}[sep=large]
      \IMode{\Dom{J_0}}
      \arrow[r, "\OMode{\phi}"]
      \arrow[rrrr, dashed, bend left=35, "\displaystyle\IMode{\SEQ(\phi,\Ul{\psi})}"]
      &
      \IMode{\infty \Dom{\State(J_1)}}
      \arrow[r, "\OMode{\infty(\Ul{\psi})}"]
      &
      \IMode{\infty\infty\Dom{\State(\State(J_2))}}
      \arrow[r, "\OMode{\mu_\infty}"]
      &
      \IMode{\infty \Dom{\State(\State(J_2))}}
      \arrow[r, "\OMode{\infty\Dom{\mu}}"]
      &
      \IMode{\infty\Dom{\State(J_2)}}
    \end{tikzcd}
  \]
  \caption{The generic sequencing tactical displayed as a composite.}
  \label{fig:seq}
\end{figure*}
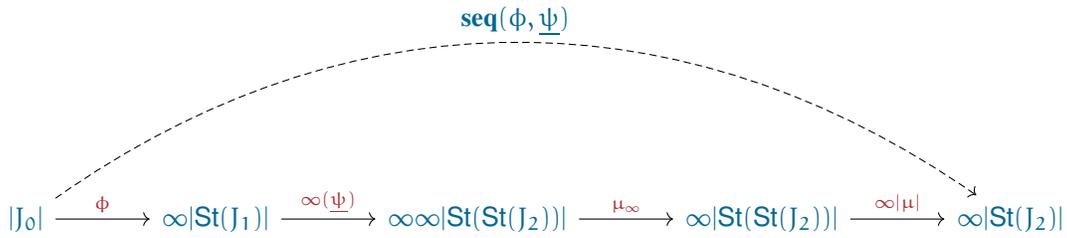

Now observe that the \LCF{} sequencing tacticals can be defined in terms
of the above combinators:
\begin{align*}
  \ADefine{
    \THEN(\phi,\psi)
  }{
    \SEQ(\phi, \ALL(\psi))
  }
  \\
  \ADefine{
    \THENL(\phi,\vec{\psi})
  }{
    \SEQ(\phi, \EACH(\vec{\psi}))
  }
\end{align*}

\subsubsection{Recursive Tacticals}

Capretta's delay monad (Figure~\ref{fig:delay-monad}) allows us to
develop a fixed point combinator for tacticals; in particular, given a
tactical $\Of{T}{\Tactics{J}{J}\to\Tactics{J}{J}}$, we have a fixed
point $\Of{\Fix(T)}{\Tactics{J}{J}}$. For the full construction of the
fixed point $\Fix(T)$, see Appendix~\ref{appendix:fixed-points}.

Using this, we can develop the standard $\REPEAT$ tactical from LCF,
which is in practice the most commonly used recursive tactical:
\begin{align*}
  \AOf{
    \REPEAT(\phi)
  }{
    \Tactics{J}{J}
  }
  \\
  \ADefine{
    \REPEAT(\phi)
  }{
    \Fix(\psi\mapsto \TRY(\THEN(\phi,\psi)))
  }
\end{align*}
Other recursive tacticals are possible, including recursive
multitacticals.

\begin{example}[Tactic for cost dynamics]
  Returning to our running example
  (Examples~\ref{ex:example-judgment-structure},
  \ref{ex:example-refinement-rules}), we can now define a useful
  tactic to discharge all $\JArith$-judgments; our first cut can be
  defined in the following way:
  \begin{align*}
    \AOf{
      \TacAutoAux
    }{
      \Tactics{\JArith}{\JArith}
    }
    \\
    \ADefine{
      \TacAutoAux
    }{
      \Now{\RuleNumEval} \oplus \Now{\RulePlusEval} \oplus \Now{\RuleAdd}
    }
    \\[6pt]
    \AOf{
      \TacAuto
    }{
      \Tactics{\JArith}{\JArith}
    }
    \\
    \ADefine{
      \TacAuto
    }{
      \REPEAT(\TacAutoAux)
    }
    \tag{*}
  \end{align*}

  This is not quite, however, what we want: the force of this tactic
  is to run all our rules repeatedly (until failure or completion) on
  each subgoal. This is fine, but because these processes are taking
  place independently on each subgoal, the instantiations induced in
  one subgoal will not propagate to an adjacent subgoal until the
  entire process has quiesced.

  The practical result of this approach is that the $\TacAuto$ tactic
  will terminate with unresolved subgoals, and must be run again; our
  intention was, however, for the tactic to discharge all subgoals
  through repetition.

  What we defined above can be described as depth-first
    repetition; what we want is breadth-first repetition, in
  which we run all the rules once on each subgoal,
  repeatedly. Then, substitutions propagate along the subgoals
  telescope with \emph{every} application of $\TacAutoAux$, instead of
  propagating only after \emph{all} applications of $\TacAutoAux$.

  The way to achieve this is to apply our repetition at the level of
  multitactics, instantiating the tactical as
  $\Of{\REPEAT}{\Tactics{\State(\JArith)}{\State(\JArith)}\to\Tactics{\State(\JArith)}{\State(\JArith)}}$
  instead of
  $\Of{\REPEAT}{\Tactics{\JArith}{\JArith}\to\Tactics{\JArith}{\JArith}}$. This
  we can accomplish as follows:
  \begin{align*}
    \AOf{
      \TacAutoMulti
    }{
      \Tactics{\State(\JArith)}{\State(\JArith)}
    }
    \\
    \ADefine{
      \TacAutoMulti
    }{
      \REPEAT(\ALL(\TacAutoAux))
    }
    \\[6pt]
    \AOf{
      \TacAuto
    }{
      \Tactics{\JArith}{\JArith}
    }
    \\
    \ADefine{
      \TacAuto
    }{
      \SEQ(\ID, \TacAutoMulti)
    }
  \end{align*}
\end{example}

\section{Concrete Implementation in \StandardML{}}

As part of the \RedPRL{} project~\cite{redprl:2016}, we have built a
practical implementation of the apparatus described above in the
\StandardML{} programming
language~\cite{milner-tofte-harper-macqueen:1997}.\footnote{Our
  implementation is open-source and available at
  \url{http://www.github.com/RedPRL/sml-dependent-lcf}.}

\RedPRL{} is an interactive proof assistant in the Nuprl tradition for
computational cubical type
theory~\cite{angiuli-harper-wilson-popl:2016}, a higher dimensional
variant of Martin-L\"of's extensional type
theory~\cite{martin-lof:1979}. Replacing \LCF{} with
\DependentLCF{} has enabled us to eliminate every last disruption to
the proof refinement process in \RedPRL{}'s refinement logic,
including the introduction rule for dependent sums (as described in
Section~\ref{sec:barbarism}), and dually, the elimination rule for
dependent products.

\DependentLCF{} has also sufficed for us as a matrix in which to
develop sophisticated type synthesis rules \`a la bidirectional
typing, which has greatly simplified the proof obligations routinely
incurred in an elaborator or refiner for extensional type theory,
without needing to develop brittle and complex tactics for this
purpose as was required in the Nuprl System.

Our experience suggests that, contrary to popularly-accepted folk
wisdom, practical and usable implementations of extensional type
theory are eminently possible, assuming that a powerful enough form of
proof refinement apparatus is adopted.

\section*{Acknowledgments}

The first author would like to thank David Christiansen for many hours of
discussion on dependent proof refinement, and Sam Tobin-Hochstadt for
graciously funding a visit to Indiana University during which the seeds for
this paper were planted. Thanks also to Arnaud Spiwack, Adrien Guatto, Danny
Gratzer, Brandon Bohrer and Darin Morrison for helpful conversations about
proof refinement.
The authors gratefully acknowledge the support of the Air Force Office of
Scientific Research through MURI grant FA9550-15-1-0053. Any opinions, findings
and conclusions or recommendations expressed in this material are those of the
authors and do not necessarily reflect the views of the AFOSR.

\bibliographystyle{IEEEtran}
\bibliography{IEEEabrv,references/refs}

\appendix

\subsection{Formal Definitions}
\label{appendix:formal-defs}

We have used a convenient syntactic notation based on variable binding
in the above. The definitions of telescopes and proof states can be
given more precisely in category-theoretic notation, at the cost of
some bureaucracy. In return, we can tell that these objects are indeed
judgment structures according to our definition purely on the basis of
how they are formed.

\begin{definition}[Yoneda embedding]
  Write $\Of{\Yoneda{\Gamma}}{\Psh{\Th}}$ for the
  representable presheaf $\OMode{\ThHom{-}{\Gamma}}$ of $\Gamma$-terms
  in the \enquote{current} context.
\end{definition}

The collection of terms-in-context can be internalized into our
presheaf category using the Yoneda embedding and the
exponential. Define the presheaf
$\Of{(\Gamma\vdash\Delta)}{\Psh{\Th}}$ as the exponential
$\OMode{{\Yoneda{\Delta}}^{\Yoneda{\Gamma}}}$; recall the definition
of the exponential in a functor category:
\begin{align*}
  \AIsEq{(\Gamma\vdash\Delta)(\Xi)}{
    {\Yoneda{\Delta}}^{\Yoneda{\Gamma}}(\Xi)
  }
  \\
  \AIsEquiv{}{
    \left[
      \Yoneda{\Xi}\times\Yoneda{\Gamma}, \Yoneda{\Delta}
    \right]
  }
  \tag{definition}
  \\
  \AIsEquiv{}{
    \left[
      \Yoneda{\Xi\times\Gamma},\Yoneda{\Delta}
    \right]
  }
  \tag{limit preservation}
  \\
  \AIsEquiv{}{
    \Yoneda{\Delta}(\Xi\times\Gamma)
  }
  \tag{Yoneda lemma}
  \\
  \AIsEquiv{}{
    \ThHom{\Xi\times\Gamma}{\Delta}
  }\tag{definition}
\end{align*}

In what follows, we will frequently exponentiate by a representable
functor; as above, naturality guarantees that the resulting object
depends in no essential way on its input, justifying a purely
syntactic notation based on variable binding.

Our definitions of telescopes and proof states as judgment structures
can be restated in categorical form below; let $\mathsf{p}$ be the evident
variable projection map in $\Th$.

\begin{align*}
  \ADefine{
    \Tele(J)
  }{
    \textstyle\mu T.\ \One + \sum_{X:J} T^{\Yoneda{\JProj{J}(X)}}
  }
  \\
  \ADefine{
    \JProj{\Tele(J)}(\mathsf{roll}(\mathsf{inl}(*)))
  }{
    \One
  }
  \\
  \ADefine{
    \JProj{\Tele(J)}(\mathsf{roll}(\mathsf{inr}(X, \Psi)))
  }{
    \JProj{J}(X) \times \JProj{\Tele(J)}(\Psi(\mathsf{p}))
  }
\end{align*}

\begin{align*}
  \ADefine{
    \State(J)
  }{
    \textstyle\sum_{\Delta:\CtxsPsh}
    \left(
    \MkSet{\Fail,\bot} +
    \sum_{\Psi:\Tele(J)}
    \JProj{\Tele(J)}(\Psi)\vdash\Delta
    \right)
  }
  \\
  \ADefine{
    \JProj{\State(J)}(\Delta,\dots)
  }{
    \Delta
  }
\end{align*}

\subsection{Proofs of Theorems}\label{appendix:proofs}

\begin{lem}\label{lem:yank-prefix}
  We have the following identity for $\Of{S}{\State(J)}$,
  $\Of{\Psi}{\Tele(J)}$, $\Of{\Ul{\Psi}}{\Tele(\State(J))}$ and
  $\Of{\mathfrak{a}}{\JProj{\Tele(\State(J))}(\TCons{x}{\WkSt{\Psi}{S}}{\Ul{\Psi}}) \vdash\Delta}$:
  \[
    \IsEq{
      \StMul{}{
        \MkState{
          \TCons{x}{
            (\WkSt{\Psi}{S})
          }{\Ul{\Psi}}
        }{\mathfrak{a}}
      }
    }{
      \WkSt{\Psi}{
        \StMul{}{
          \MkState{
            \TCons{x}{S}{\Ul{\Psi}}
          }{
            \mathfrak{a}
          }
        }
      }
    }
  \]
\end{lem}
\begin{proof}
  We proceed by case on $S$.

  \begin{proofcases}

  \item[$\Match{S}{\Fail}$] By calculation.
    \begin{align*}
      \AIsEq{
        \StMul{}{
          \MkState{
            \TCons{x}{(\WkSt{\Psi}{S})}{\Ul{\Psi}}
          }{\mathfrak{a}}
        }
      }{
        \StMul{}{
          \MkState{\TCons{x}{(\WkSt{\Psi}{\Fail})}{\Ul{\Psi}}}{\mathfrak{a}}
        }
      }
      \\
      \AIsEq{}{
        \StMul{}{\MkState{\TCons{x}{\Fail}{\Ul{\Psi}}}{\mathfrak{a}}}
      }
      \\
      \AIsEq{}{\Fail}
    \end{align*}
    \begin{align*}
      \AIsEq{
        \WkSt{\Psi}{
          \StMul{}{
            \MkState{
              \TCons{x}{S}{\Ul{\Psi}}
            }{
              \mathfrak{a}
            }
          }
        }
      }{
        \WkSt{\Psi}{
          \StMul{}{
            \MkState{
              \TCons{x}{\Fail}{\Ul{\Psi}}
            }{
              \mathfrak{a}
            }
          }
        }
      }
      \\
      \AIsEq{}{
        \WkSt{\Psi}{\Fail}
      }
      \\
      \AIsEq{}{\Fail}
    \end{align*}

  \item[$\Match{S}{\bot}$] Analogous to the above.

  \item[$\Match{S}{\MkState{\Psi_S}{\mathfrak{a}_S}}$]
    By calculation.
  \begin{align*}
    \AIsEq{
      \StMul{}{
        \MkState{
          \TCons{x}{
            (\WkSt{\Psi}{S})
          }{\Ul{\Psi}}
        }{\mathfrak{a}}
      }
    }{
      \StMul{}{
        \MkState{
          \TCons{x}{
            (\WkSt{\Psi}{(\MkState{\Psi_S}{\mathfrak{a}_S})})
          }{\Ul{\Psi}}
        }{\mathfrak{a}}
      }
    }
    \\
    \AIsEq{}{
      \StMul{}{
        \MkState{
          \TCons{x}{
            (\MkState{(\TConcat{\Psi}{\Psi_S})}{\mathfrak{a}_S})
          }{\Ul{\Psi}}
        }{\mathfrak{a}}
      }
    }
    \\
    \AIsEq{}{
      \WkSt{
        (\TConcat{\Psi}{\Psi_S})
      }{
        \StMul{}{
          \MkState{\Ul{\Psi}}{\mathfrak{a}}
        }[\mathfrak{a}_S/x]
      }
    }
  \end{align*}
  \begin{align*}
    \AIsEq{
      \WkSt{\Psi}{
        \StMul{}{
          \MkState{
            \TCons{x}{S}{\Ul{\Psi}}
          }{
            \mathfrak{a}
          }
        }
      }
    }{
      \WkSt{\Psi}{
        \StMul{}{
          \MkState{
            \TCons{x}{(\MkState{\Psi_S}{\mathfrak{a}_S})}{\Ul{\Psi}}
          }{
            \mathfrak{a}
          }
        }
      }
    }
    \\
    \AIsEq{}{
      \WkSt{\Psi}{
        (\WkSt{\Psi_S}{
          \StMul{}{
            \MkState{
              \Ul{\Psi}
            }{
              \mathfrak{a}
            }
          }[\mathfrak{a}_S/x]
        })
      }
    }
    \\
    \AIsEq{}{
    \WkSt{(\TConcat{\Psi}{\Psi_S})}{
        \StMul{}{
          \MkState{
            \Ul{\Psi}
          }{
            \mathfrak{a}
          }
        }[\mathfrak{a}_S/x]
      }
    }
  \end{align*}
  \end{proofcases}

\end{proof}

\CoherenceLemma*

\begin{proof}
  First, we have to show that the left triangle of the unit identity
  commutes for any $\Match{S}{\MkState{\Psi}{\mathfrak{a}}}$:
  \begin{align*}
    \AIsEq{
      \StMul{}{\StUnit{}{S}}
    }{
      \StMul{}{\MkState{\TCons{x}{\MkState{\Psi}{\mathfrak{a}}}{\TNil}}{x}}
    }
    \\
    \AIsEq{}{
      \MkState{\Psi}{x[\mathfrak{a}/x]}
    }
    \\
    \AIsEq{}{
      \MkState{\Psi}{\mathfrak{a}}
    }
    \\
    \AIsEq{}{
      S
    }
  \end{align*}

  Next, we must show that the right unit triangle commutes. This, we will do by induction on $\Psi$:
  \begin{proofcases}

  \item[$\Match{\Psi}{\TNil}$]
    \begin{align*}
      \AIsEq{
        \StMul{}{
          \State(\eta)(S)
        }
      }{
        \StMul{}{
          \State(\eta)(\MkState{\TNil}{\mathfrak{a}})
        }
      }
      \\
      \AIsEq{}{
        \StMul{}{
          \MkState{\TNil}{\mathfrak{a}}
        }
      }
      \\
      \AIsEq{}{
        \MkState{\TNil}{\mathfrak{a}}
      }
      \\
      \AIsEq{}{S}
    \end{align*}

  \item[$\Match{\Psi}{\TCons{x}{X}{\Psi'}}$]
    \begin{align*}
      \AIsEq{
        \StMul{}{\State(\eta)(S)}
      }{
        \StMul{}{\State(\eta)(\MkState{\TCons{x}{X}{\Psi'}}{\mathfrak{a}})}
      }
      \\
      \AIsEq{}{
        \StMul{}{
          \MkState{
            \Tele(\eta)(\TCons{x}{X}{\Psi'})
          }{\mathfrak{a}}
        }
      }
      \\
      \AIsEq{}{
        \StMul{}{
          \MkState{
            \TCons{x}{\StUnit{}{X}}{\Tele(\eta)(\Psi')}
          }{\mathfrak{a}}
        }
      }
      \\
      \AIsEq{}{
        \StMul{}{
          \MkState{
            \TCons{x}{(\MkState{\TCons{y}{X}{\TNil}}{y})}{\Tele(\eta)(\Psi')}
          }{\mathfrak{a}}
        }
      }
      \\
      \AIsEq{}{
        \WkSt{
          (\TCons{y}{X}{\TNil})
        }{
          \StMul{}{
            \MkState{\Tele(\eta)(\Psi')}{\mathfrak{a}}
          }[y/x]
        }
      }
      \\
      \AIsEq{}{
        \WkSt{
          (\TCons{x}{X}{\TNil})
        }{
          \StMul{}{
            \MkState{\Tele(\eta)(\Psi')}{\mathfrak{a}}
          }
        }
      }\tag{ren.}
      \\
      \AIsEq{}{
        \WkSt{
          (\TCons{x}{X}{\TNil})
        }{
          \StMul{}{
            \State(\eta)(\MkState{\Psi'}{\mathfrak{a}})
          }
        }
      }
      \\
      \AIsEq{}{
        \WkSt{
          (\TCons{x}{X}{\TNil})
        }{
          \MkState{\Psi'}{\mathfrak{a}}
        }
      }\tag{i.h.}
      \\
      \AIsEq{
      }{
        \MkState{\TCons{x}{X}{\Psi'}}{\mathfrak{a}}
      }
      \\
      \AIsEq{}{S}
    \end{align*}

  \end{proofcases}

  Lastly, we need to show that the monad multiplication square
  commutes. Fix $\Of{\DbUl{S}}{\State(\State(\State(J)))}$ and proceed
  by case.

  \begin{proofcases}
    \item[$\Match{\DbUl{S}}{\Fail}$]
      \begin{align*}
        \AIsEq{
          \StMul{}{\StMul{}{\DbUl{S}}}
        }{
          \StMul{}{\Fail}
        }
        \\
        \AIsEq{
          \StMul{}{\State(\mu)(\DbUl{S})}
        }{
          \StMul{}{\Fail}
        }
      \end{align*}

    \item[$\Match{\DbUl{S}}{\bot}$] Analogous to the above.

    \item[$\Match{\DbUl{S}}{\MkState{\DbUl{\Psi}}{\mathfrak{a}}}$]
      Proceed by induction on $\DbUl{\Psi}$.
  \end{proofcases}

  \begin{proofcases}

  \item[$\Match{\DbUl{\Psi}}{\TNil}$]
    \begin{align*}
      \AIsEq{
        \StMul{}{\StMul{}{\DbUl{S}}}
      }{
        \StMul{}{\StMul{}{\MkState{\TNil}{\mathfrak{a}}}}
      }
      \\
      \AIsEq{}{
        \StMul{}{\MkState{\TNil}{\mathfrak{a}}}
      }
    \end{align*}
    \begin{align*}
      \AIsEq{
        \StMul{}{\State(\mu)(\DbUl{S})}
      }{
        \StMul{}{\State(\mu)(\MkState{\TNil}{\mathfrak{a}})}
      }
      \\
      \AIsEq{}{
        \StMul{}{\MkState{\TNil}{\mathfrak{a}}}
      }
    \end{align*}

  \item[$\Match{\DbUl{\Psi}}{\TCons{x}{\Ul{S}}{\DbUl{\Psi'}}}$ with $\Match{\Ul{S}}{\Fail}$]
    \begin{align*}
      \AIsEq{
        \StMul{}{\StMul{}{\DbUl{S}}}
      }{
        \StMul{}{\StMul{}{\MkState{\TCons{x}{\Ul{S}}{\DbUl{\Psi'}}}{\mathfrak{a}}}}
      }
      \\
      \AIsEq{}{
        \StMul{}{\Fail}
      }
      \\
      \AIsEq{}{
        \Fail
      }
    \end{align*}
    \begin{align*}
      \AIsEq{
        \StMul{}{\State(\mu)(\DbUl{S})}
      }{
        \StMul{}{\State(\mu)(\MkState{\TCons{x}{\Ul{S}}{\DbUl{\Psi'}}}{\mathfrak{a}})}
      }
      \\
      \AIsEq{}{
        \StMul{}{\MkState{\TCons{x}{\Fail}{\Tele(\mu)(\DbUl{\Psi'})}}{\mathfrak{a}}}
      }
      \\
      \AIsEq{}{
        \Fail
      }
    \end{align*}

  \item[$\Match{\DbUl{\Psi}}{\TCons{x}{\Ul{S}}{\DbUl{\Psi'}}}$ with $\Match{\Ul{S}}{\bot}$]
    Analogous to the above.

  \item[$\Match{\DbUl{\Psi}}{\TCons{x}{\Ul{S}}{\DbUl{\Psi'}}}$ with $\Match{\Ul{S}}{\MkState{\Ul{\Psi}}{\mathfrak{b}}}$]
    \begin{align*}
      \AIsEq{
        \StMul{}{\StMul{}{\DbUl{S}}}
      }{
        \StMul{}{\StMul{}{\MkState{\TCons{x}{\Ul{S}}{\DbUl{\Psi'}}}{\mathfrak{a}}}}
      }
      \\
      \AIsEq{}{
        \StMul{}{\StMul{}{\MkState{\TCons{x}{(\MkState{\Ul{\Psi}}{\mathfrak{b}})}{\DbUl{\Psi'}}}{\mathfrak{a}}}}
      }
      \\
      \AIsEq{}{
        \StMul{}{
          \WkSt{\Ul{\Psi}}{
            \StMul{}{
              \MkState{\DbUl{\Psi'}}{\mathfrak{a}}
            }[\mathfrak{b}/x]
          }
        }
      }
    \end{align*}
    \begin{align*}
      \AIsEq{
        \StMul{}{\State(\mu)(\DbUl{S})}
      }{
        \StMul{}{\State(\mu)(\MkState{\TCons{x}{\Ul{S}}{\DbUl{\Psi'}}}{\mathfrak{a}})}
      }
      \\
      \AIsEq{}{
        \StMul{}{
          \MkState{
            \TCons{x}{\StMul{}{\Ul{S}}}{\Tele(\mu)(\DbUl{\Psi'})}
          }{\mathfrak{a}}
        }
      }
      \\
      \AIsEq{}{
        \StMul{}{
          \MkState{
            \TCons{x}{\StMul{}{\MkState{\Ul{\Psi}}{\mathfrak{b}}}}{\Tele(\mu)(\DbUl{\Psi'})}
          }{\mathfrak{a}}
        }
      }
    \end{align*}

    To proceed further, we must perform a second induction on $\Ul{\Psi}$.
    \begin{proofcases}

    \item[$\Match{\Ul{\Psi}}{\TNil}$]
      \begin{align*}
        \AIsEq{
          \StMul{}{\StMul{}{\DbUl{S}}}
        }{
          \StMul{}{
            \StMul{}{
              \MkState{\DbUl{\Psi'}}{\mathfrak{a}}
            }
          }[\mathfrak{b}/x]
        }
      \end{align*}
      \begin{align*}
        \AIsEq{
          \StMul{}{\State(\mu)(\DbUl{S})}
        }{
          \StMul{}{
            \MkState{
              \TCons{x}{\MkState{\TNil}{\mathfrak{b}}}{\Tele(\mu)(\DbUl{\Psi'})}
            }{\mathfrak{a}}
          }
        }
        \\
        \AIsEq{}{
          \StMul{}{\MkState{\Tele(\mu)(\DbUl{\Psi'})}{\mathfrak{a}}}[\mathfrak{b}/x]
        }
        \\
        \AIsEq{}{
          \StMul{}{\State(\mu)(\MkState{\DbUl{\Psi'}}{\mathfrak{a}})}[\mathfrak{b}/x]
        }
      \end{align*}

      By the outer inductive hypothesis, we have
      $\IsEq{\StMul{}{\StMul{}{\MkState{\DbUl{\Psi'}}{\mathfrak{a}}}}}{\StMul{}{\State(\mu)(\MkState{\DbUl{\Psi'}}{\mathfrak{a}})}}$.

      \item[$\Match{\Ul{\Psi}}{\TCons{y}{S}{\Ul{\Psi'}}}$ with $\Match{S}{\Fail}$]
        \begin{align*}
          \AIsEq{
            \StMul{}{\StMul{}{\DbUl{S}}}
          }{
            \StMul{}{
              \WkSt{
                \TCons{y}{\Fail}{\Ul{\Psi'}}
              }{
                \StMul{}{
                  \MkState{\DbUl{\Psi'}}{\mathfrak{a}}
                }[\mathfrak{b}/x]
              }
            }
          }
          \\
          \AIsEq{}{\Fail}
        \end{align*}
        \begin{align*}
          \AIsEq{
            \StMul{}{\State(\mu)(\DbUl{S})}
          }{
            \StMul{}{
              \MkState{
                \TCons{x}{\StMul{}{\MkState{\TCons{y}{\Fail}{\Ul{\Psi'}}}{\mathfrak{b}}}}{\Tele(\mu)(\DbUl{\Psi'})}
              }{\mathfrak{a}}
            }
          }
          \\
          \AIsEq{}{
            \StMul{}{
              \MkState{
                \TCons{x}{\Fail}{\Tele(\mu)(\DbUl{\Psi'})}
              }{\mathfrak{a}}
            }
          }
          \\
          \AIsEq{}{\Fail}
        \end{align*}

      \item[$\Match{\Ul{\Psi}}{\TCons{y}{S}{\Ul{\Psi'}}}$ with $\Match{S}{\bot}$]
        Analogous to the above.

      \item[$\Match{\Ul{\Psi}}{\TCons{y}{S}{\Ul{\Psi'}}}$ with $\Match{S}{\MkState{\Psi}{\mathfrak{c}}}$]
        \begin{align*}
          \AIsEq{
            \StMul{}{\StMul{}{\DbUl{S}}}
          }{
            \StMul{}{
              \WkSt{
                \TCons{y}{(\MkState{\Psi}{\mathfrak{c}})}{\Ul{\Psi'}}
              }{
                \StMul{}{
                  \MkState{\DbUl{\Psi'}}{\mathfrak{a}}
                }[\mathfrak{b}/x]
              }
            }
          }
          \\
          \AIsEq{}{
            \WkSt{\Psi}{
              \StMul{}{(
                \WkSt{
                  \Ul{\Psi'}
                }{
                  \StMul{}{\MkState{\DbUl{\Psi'}}{\mathfrak{a}}}[\mathfrak{b}/x]
                }
              )}[\mathfrak{c}/y]
            }
          }
        \end{align*}
        \begin{align*}
          \AIsEq{
            \StMul{}{\State(\mu)(\DbUl{S})}
          }{
            \StMul{}{
              \MkState{
                \TCons{x}{\StMul{}{\MkState{\TCons{y}{(\MkState{\Psi}{\mathfrak{c}})}{\Ul{\Psi'}}}{\mathfrak{b}}}}{\Tele(\mu)(\DbUl{\Psi'})}
              }{\mathfrak{a}}
            }
          }
          \\
          \AIsEq{}{
            \StMul{}{
              \MkState{
                \TCons{x}{
                  (\WkSt{\Psi}{
                    \StMul{}{\MkState{\Ul{\Psi'}}{\mathfrak{b}}}[\mathfrak{c}/y]
                  })
                }{
                  \Tele(\mu)(\DbUl{\Psi'})
                }
              }{\mathfrak{a}}
            }
          }
          \\
          \AIsEq{}{
            \WkSt{\Psi}{
              \StMul{}{
                \MkState{
                  \TCons{x}{
                     \StMul{}{\MkState{\Ul{\Psi'}}{\mathfrak{b}}}[\mathfrak{c}/y]
                  }{
                    \Tele(\mu)(\DbUl{\Psi'})
                  }
                }{\mathfrak{a}}
              }
            }
          }
          \tag{Lemma~\ref{lem:yank-prefix}}
        \end{align*}
        Now, we need to show the following:
        \[
          \begin{array}{l}
          \IMode{
            \StMul{}{(
              \WkSt{
                \Ul{\Psi'}
              }{
                \StMul{}{\MkState{\DbUl{\Psi'}}{\mathfrak{a}}}[\mathfrak{b}/x]
              }
            )}[\mathfrak{c}/y]
          }
          \\ \qquad =
          \IMode{
            \StMul{}{
              \MkState{
                \TCons{x}{
                   \StMul{}{\MkState{\Ul{\Psi'}}{\mathfrak{b}}}[\mathfrak{c}/y]
                }{
                  \Tele(\mu)(\DbUl{\Psi'})
                }
              }{\mathfrak{a}}
            }
          }
          \end{array}
        \]

        It suffices to show:
        \[
          \begin{array}{l}
          \IMode{
            \StMul{}{(
              \WkSt{
                \Ul{\Psi'}
              }{
                \StMul{}{\MkState{\DbUl{\Psi'}}{\mathfrak{a}}}[\mathfrak{b}/x]
              }
            )}[\mathfrak{c}/y]
          }
          \\ \qquad =
          \IMode{
            \StMul{}{
              \MkState{
                \TCons{x}{
                   \StMul{}{\MkState{\Ul{\Psi'}}{\mathfrak{b}}}
                }{
                  \Tele(\mu)(\DbUl{\Psi'})
                }
              }{\mathfrak{a}}
            }[\mathfrak{c}/y]
          }
          \end{array}
        \]
        Cancelling substitutions, we need only show the following:
        \[
          \begin{array}{l}
          \IMode{
            \StMul{}{(
              \WkSt{
                \Ul{\Psi'}
              }{
                \StMul{}{\MkState{\DbUl{\Psi'}}{\mathfrak{a}}}[\mathfrak{b}/x]
              }
            )}
          }\\ \qquad =
          \IMode{
            \StMul{}{
              \MkState{
                \TCons{x}{
                   \StMul{}{\MkState{\Ul{\Psi'}}{\mathfrak{b}}}
                }{
                  \Tele(\mu)(\DbUl{\Psi'})
                }
              }{\mathfrak{a}}
            }
          }
          \end{array}
        \]
        But this holds by our inner inductive hypothesis.
    \end{proofcases}
  \end{proofcases}

\end{proof}

\subsection{Fixed Points of Tacticals}
\label{appendix:fixed-points}

We will now demonstrate how to take the fixed point of a tactical
using Capretta's delay monad (Figure~\ref{fig:delay-monad}), and use
it to construct the commonly-used $\REPEAT$ tactical from \LCF{}. To
begin with, we will need to define fiberwise products and
$\omega$-sequences as operations on judgment structures. The fiberwise
product of two judgment structures is a judgment structure whose
judgments consist in pairs of judgments that synthesize the same sort
$\Delta$:

\[
  \framebox{$\Of{\otimes}{\DiscJStr\times\DiscJStr\to\DiscJStr}$}
  \qquad
  \vcenter{
    \infer{
      \IsJdg{\Gamma}{\Tuple{X,Y}}{J_0\otimes J_1}{\Delta}
    }{
      \IsJdg{\Gamma}{X}{J_0}{\Delta}
      &
      \IsJdg{\Gamma}{Y}{J_1}{\Delta}
    }
  }
  \tag{Fiberwise Product}
\]

This is just the pullback of the two judgment structures when viewed
as objects in the slice category $\Slice{\LaxPsh{\Dom{\Th}}}{\CtxsPsh}$:
\[
  \begin{tikzcd}[sep=large]
    \OMode{J_0\otimes J_1}
    \arrow[r]
    \arrow[d]
    \arrow[dr,phantom, "\lrcorner", very near start]
    &
    \IMode{J_0}
    \arrow[d, "\IMode{\JProj{J_0}}"]
    \\
    \IMode{J_1}
    \arrow[r, "\IMode{\JProj{J_1}}"]
    &
    \IMode{\CtxsPsh}
  \end{tikzcd}
\]

Next, we define an endofunctor on judgment structures that takes
infinite sequences of judgments:
\[
  \framebox{$\Of{-^\omega}{\DiscJStr\to\DiscJStr}$}
  \qquad
  \vcenter{
    \infer{
      \IsJdg{\Gamma}{\Tuple{X_n\mid n}}{J^\omega}{\Delta}
    }{
      \IsJdg{\Gamma}{X}{J_n}{\Delta}\ (\Member{n}{\Nat})
    }
  }
  \tag{$\omega$-Sequence}
\]
This too can be presented as the limit
$\OMode{\varprojlim_{n\in\Nat} J^n}$ where $J^n$ is the
$n$-fold fiberwise product of $J$ with itself.

Now that we have defined the objects we require, we will begin to
define the operations by which we can take the fixed point of a
tactical, following~\cite{capretta:2005}. First, we have a way to
\enquote{race} two delayed judgments, returning the one which resolves
first:
\begin{align*}
  \AOf{
    \Race
  }{
    \infty J \otimes \infty J \to \infty J
  }
  \\
  \IMode{\Race_\Gamma}
  &
  \;\begin{RuleDefn}
    \RuleCase{\Tuple{\Now{X},Y_\infty}}{
      \Now{X}
    }
    \\
    \RuleCase{\Tuple{\Later{X_\infty},\Now{Y}}}{
      \Now{Y}
    }
    \\
    \RuleCase{\Tuple{\Later{X_\infty},\Later{Y_\infty}}}{
      \Race_\Gamma\Tuple{X_\infty,Y_\infty}
    }
  \end{RuleDefn}
\end{align*}
This construction can be lifted to an $\omega$-sequence of judgments
as follows:
\begin{align*}
  \AOf{
    \Search{n}
  }{
    {(\infty J)}^\omega \otimes \infty J \to \infty J
  }
  \\
  \IMode{\Search{n}^\Gamma}
  &
  \begin{RuleDefn}
    \RuleCase{\Tuple{F,\Now{X}}}{\Now{X}}
    \\
    \RuleCase{\Tuple{F,\Later{X_\infty}}}{
      \Later{\Search{n+1}^\Gamma\Tuple{F, \Race_\Gamma\Tuple{X_\infty, F_n}}}
    }
  \end{RuleDefn}
\end{align*}

In the delay monad, it is possible to define an object which never
resolves:
\begin{align*}
  \AOf{\Never_\Gamma}{\infty J(\Gamma)}
  \\
  \ADefEq{
    \Never_\Gamma
  }{
    \Later{\Never_\Gamma}
  }
\end{align*}

Now, using this and our unbounded search operator, we can take the
least upper bound of an $\omega$-sequence of judgments:
\begin{align*}
  \AOf{
    \sqcup
  }{
    {(\infty J)}^\omega \to \infty J
  }
  \\
  \ADefine{
    \sqcup_\Gamma(F)
  }{
    \Search{0}^\Gamma\Tuple{F,\Never_\Gamma}
  }
\end{align*}

Now fix a tactical $\Of{T}{\Tactics{J}{J}\to\Tactics{J}{J}}$; it is
now easy to get the fixed point of $T$ (if it exists) by taking the
least upper bound of a sequence of increasingly many applications of
$T$ to itself:
\begin{align*}
  \AOf{
    \Fix(T)
  }{
    \Tactics{J}{J}
  }
  \\
  \ADefine{
    {\Fix(T)}_\Gamma(X)
  }{
    \sqcup\Tuple{T^n_\Gamma(X) \mid n}
  }
\end{align*}
where
\begin{align*}
  \AOf{
    T^n
  }{
    \Tactics{J}{J}
  }
  \\
  \ADefEq{
    T^0_\Gamma
  }{
    X\mapsto \Never_\Gamma
  }
  \\
  \ADefEq{
    T^{n+1}_\Gamma
  }{
    T_\Gamma(T^n_\Gamma)
  }
\end{align*}

\subsection{Defining Refinement Logics}\label{appendix:defining-refinement-logics}

So far, we have built up a sophisticated apparatus for deterministic
refinement proof, but have not shown how to instantiate it to a
closed logic.  In what follows, we will define a category of
\emph{refiners}, which can be thought of as implementations of a
logic.

\begin{definition}[Heterogeneous refiners]
  A \emph{heterogeneous refiner} for judgment structures
  $\Of{J_0}{\OpCat{\JStr}}$ and $\Of{J_1}{\JStr}$ is a signature of
  rule names $\Of{\Sigma}{\POS}$ equipped with an interpretation
  $\Of{\mathcal{R}}{\Sigma\to\Rules{J_0}{J_1}}$. More
  generally, the \emph{category of heterogeneous refiners}
  $\IMode{\HRefiners(J_0,J_1)}$ is the lax slice category
  $\OMode{\LaxSlice{\POS}{\Rules{J_0}{J_1}}}$.

  That is to say, a \emph{refiner homomorphism} is a renaming of rules
  which preserves behavior up to approximation:
  \[
    \begin{tikzcd}[sep=huge]
      \IMode{\Sigma_0}
      \arrow[r, "\IMode{\phi}"]
      \arrow[dr, swap, "\IMode{\mathcal{R}_0}"]
      \arrow[dr, phantom, bend left=25,"\preccurlyeq"]
      &
      \IMode{\Sigma_1}
      \arrow[d, "\IMode{\mathcal{R}_1}"]
      \\
      &\IMode{\Rules{J_0}{J_1}}
    \end{tikzcd}
  \]

\end{definition}

Because this definition gives rise to a functor
$\Of{\HRefiners}{\OpCat{\JStr}\times\JStr\to\CAT}$, via the
Grothendieck construction we can view refiners in general as forming a
fibered category
$\MkFibration{\HRefiners}{p}{\JStr\times\OpCat{\JStr}}$, defining
$\Define{\HRefiners}{\oint^{\JStr\times\OpCat{\JStr}}\HRefiners}$. The
fibrational version has the advantage of specifying the
notion of a refiner without fixing a particular judgment structure,
thereby enabling a direct characterization of homomorphisms between
refiners over different judgment structures.

We will usually restrict our attention to refiners that
transform goals into subgoals of the same judgment structure;
therefore, we must define a notion of \emph{homogeneous refiner}.

First, let $\Core{\JStr}$ be the \emph{groupoid core} of $\JStr$,
i.e.\ the largest subcategory of $\JStr$ whose arrows are all
isomorphisms; we have the evident diagonal functor
$\Of{\delta}{\Core{\JStr}\to\JStr\times\OpCat{\JStr}}$.
Then, the category $\Refiners$ of \emph{homogeneous refiners} is
easily described as the following pullback of the refiner fibration
along the diagonal:
\[
  \begin{tikzcd}[sep=huge]
    \OMode{\Refiners}
    \arrow[r]
    \arrow[d, swap,"\OMode{\Core{p}}"]
    \arrow[dr,phantom, "\lrcorner", very near start]
    &
    \IMode{\HRefiners}
    \arrow[d, "\IMode{p}"]
    \\
    \IMode{\Core{\JStr}}
    \arrow[r, swap, "\IMode{\delta}"]
    &
    \IMode{\JStr\times\OpCat{\JStr}}
  \end{tikzcd}
\]

Because it is the pullback of a fibration,
$\MkFibration{\Refiners}{\Core{p}}{\Core{\JStr}}$ is also a fibered
category.

\newcommand\Leaf[1]{\mathsf{leaf}(#1)}
\newcommand\Branch[2]{\mathsf{branch}(#1;\;#2)}

\begin{definition}[Derivability Closure]

  We can define the derivability closure of a homogeneous refiner
  $\Of{\IBox{\Match{\mathfrak{R}}{\Tuple{J,\Sigma,\mathcal{R}}}}}{\Refiners}$.
  First define a new rule signature which contains all derived rules, and
  extend the interpreter $\mathcal{R}$ appropriately:
  \[
    \infer{
      \Member{\Leaf{r}}{\Sigma^\star}
    }{
      \Member{r}{\Sigma}
    }
    \qquad
    \infer{
      \Member{\Branch{r}{\vec{r}}}{\Sigma^\star}
    }{
      \Member{r}{\Sigma^\star}
      &
      \Member{\vec{r}}{\List{\Sigma^\star}}
    }
  \]
  \[
    \infer{
      \IsApprox{\Leaf{r_0}}{\Leaf{r_1}}
    }{
      \IsApprox{r_0}{r_1}
    }
    \qquad
    \infer{
      \IsApprox{\Branch{r_0}{\vec{r}_0}}{\Branch{r_1}{\vec{r}_1}}
    }{
      \IsApprox{r_0}{r_1}
      &
      \IsApprox{\vec{r}_0}{\vec{r}_1}
    }
  \]
  \begin{align*}
    \IMode{\mathcal{R}^\star} &: \IMode{\Sigma^\star\to\Rules{J}{J}}
    \\
    \IMode{\mathcal{R}^\star}
    &
    \;\begin{RuleDefn}
      \RuleCase{\Leaf{r}}{
        \mathcal{R}(r)
      }
      \\
      \RuleCase{\Branch{r}{\vec{r}}}{
        \mathcal{R}^*(r);\; \List{\mathcal{R}^*}(\vec{r})
      }
    \end{RuleDefn}
  \end{align*}

  Then the derivability closure of $\IMode{\mathfrak{R}}$ is
  $\OBox{\Define{\mathfrak{R}^\star}{\Tuple{J,\Sigma^\star,\mathcal{R}^\star}}}$.
\end{definition}

\begin{thm}
  We have a refiner homomorphism
  $\Of{i}{\mathfrak{R}\to\mathfrak{R}^\star}$.
\end{thm}
\begin{proof}

  Define the action on rule names as $\Define{i(r)}{\Leaf{r}}$; this is clearly
  monotone. It suffices to show that the following diagram commutes:
  \[
    \begin{tikzcd}[sep=large]
      \IMode{\Sigma}
      \arrow[r, "\IMode{i}"]
      \arrow[dr, swap, "\IMode{\mathcal{R}}"]
      &
      \IMode{\Sigma^\star}
      \arrow[d, "\IMode{\mathcal{R}^\star}"]
      \\
      &\IMode{\Rules{J}{J}}
    \end{tikzcd}
    \qquad
    \begin{tikzcd}[sep=large]
      \IMode{r}
      \arrow[r,mapsto,"\IMode{i}"]
      \arrow[dr,swap,mapsto, "\IMode{\mathcal{R}}"]
      &
      \OMode{\Leaf{r}}
      \arrow[d,mapsto,"\IMode{\mathcal{R}^\star}"]
      \\
      &\OBox{\IsEq{\mathcal{R}(r)}{\mathcal{R}(r)}}
    \end{tikzcd}
  \]

\end{proof}

\begin{definition}[Presentation]\label{def:presentation}

  A \emph{presentation} of a refiner
  $\Of{\IBox{\Define{\mathfrak{R}}{\Tuple{J,\Sigma,\mathcal{R}}}}}{\Refiners}$
  is another refiner
  $\Define{\mathfrak{R}_p}{\Tuple{J_p,\Sigma_p,\mathcal{R}_p}}$ together with a
  refiner homomorphism $\Of{p}{\mathfrak{R}_p\to\mathfrak{R}}$ such that the
  induced judgment structure homomorphism $\Of{p_0}{J_p\to J}$ is equipped with
  a section; that is, we have the following:
  \[
    \begin{tikzcd}[sep=huge]
      \IMode{J}
      \arrow[r,dashed,"\OMode{s}"]
      \arrow[dr,swap,"\IMode{\ArrId{J}}"]
      &
      \IMode{J_p}
      \arrow[d,"\IMode{p_0}"]
      \\
      &
      \IMode{J}
    \end{tikzcd}
  \]

\end{definition}

\begin{definition}[Category of Presentations]

  We can capture Definition~\ref{def:presentation} in a \emph{category of
  presentations} as the following pullback situation, letting
  $\Define{\mathfrak{R}}{\Tuple{J,\Sigma,\mathcal{R}}}$ and $\SplitEpis{\JStr}$
  be the category of split epimorphisms in $\JStr$:
  \[
    \begin{tikzcd}[sep=huge]
      \OMode{\Pres{\mathfrak{R}}}
      \arrow[r,dashed]
      \arrow[d,dashed]
      &
      \IMode{\Slice{\Refiners}{\mathfrak{R}}}
      \arrow[d,"\IMode{\pi_1}"]
      \\
      \IMode{\Slice{\SplitEpis{\JStr}}{J}}
      \arrow[r,swap,rightarrowtail,"\IMode{i}"]
      &
      \IMode{\Slice{\JStr}{J}}
    \end{tikzcd}
  \]
\end{definition}

\begin{definition}[Canonizing Presentation]
  A presentation $\Of{p}{\mathfrak{R}_p\to\mathfrak{R}}$ is
  called \emph{canonizing} when any construction that can be effected in
  $\mathfrak{R}$ can be effected by a unique rule in $\mathfrak{R}_p$. That is
  to say, for any $\IsJdg{\Gamma}{X}{J}{\Delta}$ and
  $\Of{\mathfrak{a}}{\ThHom{\Gamma}{\Delta}}$, we have:
  \[
    (\exists \Of{r}{\Sigma}.\ \Match{\mathcal{R}(r)(X)}{\MkState{\TNil}{\mathfrak{a}}})
    \Rightarrow
    (\exists! \Of{r}{\Sigma_p}.\ \Match{\mathcal{R}_p(r)(p^{-1}(X))}{\MkState{\TNil}{\mathfrak{a}}})
  \]
\end{definition}

\subsection{Tactic Scripts and their Dynamics}

Fixing a refiner
$\Of{\IBox{\Match{\mathfrak{R}}{\Tuple{J,\Sigma,\mathcal{R}}}}}{\Refiners}$,
we can now define a language of tactic scripts for $\mathfrak{R}$, letting $r$ range over $\Sigma$:
\[
  \begin{array}{lcl}
    \IMode{t} &::= &\OMode{r} \mid \OMode{1} \mid \OMode{t\oplus t} \mid \OMode{t\star} \mid \OMode{t; m}
    \\
    \IMode{m} &::= &\OMode{\Box{t}} \mid \OMode{[t,\dots,t]} \mid \OMode{m\star}%  \mid \OMode{\Foc{\bar n}{t}}
  \end{array}
\]

In the above, the classic $\IMode{\ORELSE}$ tactical from LCF is
implemented by $\IMode{t_1\oplus t_2}$; as in
Section~\ref{sec:tacticals} we have decomposed the standard LCF
tacticals $\IMode{\THEN}$ and $\IMode{\THENL}$ into a combination of
multitacticals and the sequencing tactical: respectively
$\OMode{t_0;\Box t_1}$ and $\OMode{t;[t_0,\dots,t_n]}$.
\begin{align*}
  \AOf{
    \TacInterp{t}
  }{
    \Tactics{J}{J}
  }
  \\
  \AOf{
    \MTacInterp{m}
  }{
    \Tactics{\State(J)}{\State(J)}
  }
\end{align*}
\begin{align*}
  \ADefEq{
    \TacInterp{r}
  }{
    \Now{\mathcal{R}(r)}
  }
  \\
  \ADefEq{
    \TacInterp{1}
  }{
    \ID
  }
  \\
  \ADefEq{
    \TacInterp{t_1\oplus t_2}
  }{
    \TacInterp{t_1}\oplus\TacInterp{t_2}
  }
  \\
  \ADefEq{
    \TacInterp{t\star}
  }{
    \REPEAT(\TacInterp{t})
  }
  \\
  \ADefEq{
    \TacInterp{t; m}
  }{
    \SEQ(\TacInterp{t},\TacInterp{m})
  }
\end{align*}
\begin{align*}
  \ADefEq{
    \MTacInterp{\Box t}
  }{
    \ALL(\TacInterp{t})
  }
  \\
  \ADefEq{
    \MTacInterp{[t_0,\dots,t_n]}
  }{
    \EACH(\Tuple{\TacInterp{t_0},\dots,\TacInterp{t_n}})
  }
  \\
  \ADefEq{
    \MTacInterp{m\star}
  }{
    \REPEAT(\MTacInterp{m})
  }
\end{align*}

%%% Local Variables:
%%% mode: latex
%%% TeX-master: "afpr"
%%% End:

\end{document}